%% file: main.tex
\documentclass[11pt]{article}
\usepackage{robust-recovery}
\usepackage{physics}

\allowdisplaybreaks

\renewcommand{\Pr}{\mathbf{Pr}}

\newcommand{\Model}{\mathrm{M}}
\newcommand{\SBM}{\mathrm{SBM}}

\begin{document}

\title{Robust recovery for stochastic block models, simplified and generalized}


\author{Sidhanth Mohanty\thanks{MIT. \texttt{sidhanth@csail.mit.edu}. Much of this work was done while the author was a PhD student at UC Berkeley.} \and Prasad Raghavendra\thanks{UC Berkeley. \texttt{raghavendra@berkeley.edu}.} \and David X. Wu\thanks{UC Berkeley. \texttt{david\_wu@berkeley.edu}.}}

\date{\today}
\maketitle

\input{abstract}

\setcounter{page}{0}
\thispagestyle{empty}
\newpage
\setcounter{page}{0}
\thispagestyle{empty}
\tableofcontents
\newpage

\input{intro}

\input{prelim}

\input{recovery-algo}

\input{outlier-bound}

\input{local-stat}

\input{robust-eigenspace-recovery}

\input{rounding-spectral}

\section*{Acknowledgments}
We would like to thank Omar Alrabiah and Kiril Bangachev for diligent feedback on an earlier draft of this paper.

\bibliographystyle{alpha}
\bibliography{main}
\appendix
\input{weak-correlation-definition}

\end{document}

%% file: abstract.tex
\begin{abstract}
    We study the problem of \emph{robust community recovery}: efficiently recovering communities in sparse stochastic block models in the presence of adversarial corruptions.
    In the absence of adversarial corruptions, there are efficient algorithms when the \emph{signal-to-noise ratio} exceeds the \emph{Kesten--Stigum (KS) threshold}, widely believed to be the computational threshold for this problem.
    The question we study is: \emph{does the computational threshold for robust community recovery also lie at the KS threshold?}
    We answer this question affirmatively, providing an algorithm for robust community recovery for arbitrary stochastic block models on any constant number of communities, generalizing the work of Ding, d'Orsi, Nasser \& Steurer \cite{DdNS22} on an efficient algorithm above the KS threshold in the case of $2$-community block models.

    There are three main ingredients to our work:

    \begin{enumerate}
    \item The Bethe Hessian of the graph is defined as $H_G(t) \triangleq (D_G-I)t^2 - A_Gt + I$ where $D_G$ is the diagonal matrix of degrees and $A_G$ is the adjacency matrix.
    Empirical work suggested that the Bethe Hessian for the stochastic block model has outlier eigenvectors corresponding to the communities right above the Kesten-Stigum threshold \cite{KMMNSZZ13,SKZ14}.

     We formally confirm the existence of outlier eigenvalues for the Bethe Hessian, by explicitly constructing outlier eigenvectors from the community vectors.
    \item We develop an algorithm for a variant of robust PCA on sparse matrices.  Specifically, an  algorithm to partially recover top eigenspaces from adversarially corrupted sparse matrices under mild delocalization constraints. 

    \item A rounding algorithm to turn vector assignments of vertices into a community assignment, inspired by the algorithm of Charikar \& Wirth \cite{CW04} for $2$XOR.
    \end{enumerate}


\end{abstract}

%% file: intro.tex
\section{Introduction}
The stochastic block model (SBM) has provided an enlightening lens into understanding a wide range of computational phenomena in Bayesian inference problems, such as computational phase transitions \& information-computation gaps \cite{DKMZ11,MNS18,Mas14,HS17}, spectral methods for sparse matrices \cite{KMMNSZZ13,SKZ14,BLM15}, local message-passing algorithms \cite{MNS14,PY23,GP23}, and robustness \cite{MPW16,BMR21,DdNS22,LM22}.

The SBM is a model of random graphs where the vertices are partitioned into \emph{communities}, denoted by $\bx$, and the probability of an edge existing is contingent on the communities that the two endpoints are part of.
The algorithmic task is the \emph{community recovery} problem: given an input graph $\bG$, estimate the posterior $\bx|\bG$ with an efficient algorithm.
\begin{definition}[Informal]
    In the \emph{stochastic block model}, we are given a $k\times k$ matrix $\Model$, a distribution $\pi$ over $[k]$, and $d > 0$, and $\SBM_n(\Model, \pi, d)$ denotes the distribution where an $n$-vertex graph $\bG$ is sampled by:
    \begin{enumerate}
    \item drawing a \emph{color} $\bx(u)\sim\pi$ for every $u\in[n]$,
    \item for each pair of vertices $u,v$, the edge $\{u,v\}$ is chosen with probability $\Model_{\bx(u), \bx(v)}\cdot\frac{d}{n}$ .
    \end{enumerate}
    In the \emph{community recovery} problem, the goal is to give an efficient algorithm that takes $\bG$ as input and outputs a community assignment $\wh{x}$ approximating $\bx|\bG$ (see \Cref{def:weak-recovery} for a formal definition).
\end{definition}
\noindent {\bf Computational thresholds.}
For a given $\Model$ and $\pi$, increasing $d$ can only possibly make the problem easier.
The main question is to understand the \emph{computational threshold} for community recovery --- i.e. the minimum value of $d$ where the problem goes from being intractable to admitting efficient algorithms.

The first predictions for this computational threshold came from the \emph{cavity method} in statistical physics in the work of Decelle, Krzakala, Moore \& Zdeborova \cite{DKMZ11}.
They posited that the location of this transition is at the \emph{Kesten--Stigum threshold} (henceforth KS threshold), a threshold for broadcast processes on trees studied in the works of Kesten \& Stigum \cite{KS66, KS67}.
The algorithmic side of these predictions was confirmed in the case of the $2$-community block model in the works of Mossel, Neeman \& Sly \cite{MNS18} and Massouli{\'e} \cite{Mas14}, and then for block models in increasing levels of generality by Bordenave, Lelarge \& Massouli{\'e} \cite{BLM15}, Abbe \& Sandon \cite{AS15}, and Hopkins \& Steurer \cite{HS17}.

\medskip

\noindent {\bf Robust algorithms.}
All of these algorithms utilize the knowledge of the distribution the input is sampled from quite strongly --- they are based on $\Omega(\log n)$-length walk statistics in the stochastic block model.
However, the full generative process in inference is not always known precisely.
Thus, we would like algorithms that utilize but do not overfit to the distributional assumptions.

Demanding that our algorithm be \emph{robust}, i.e. resilient to adversarial corruptions to the input, is often a useful way to design algorithms that are less sensitive to distributional assumptions.
This leads one to wonder: can algorithms that don't strongly exploit the prior distribution approach the KS threshold?

\medskip

\noindent {\bf Optimization vs.\ inference.}
Earlier approaches to robust recovery in $2$-community block models were based on \emph{optimization}: semidefinite programming relaxations of the minimum bisection problem, as in the work of Guedon \& Vershynin \cite{GV16}.
These approaches have the advantage of being naturally robust, since the algorithms are approximately Lipschitz around random inputs, but the minimum bisection relaxation is not known to achieve statistical optimality and only succeeds well above the KS threshold.

The following two results point to the suboptimality of optimization-based strategies.
Moitra, Perry \& Wein \cite{MPW16} considered the \emph{monotone} adversary in the $2$-community setting, where the adversary is allowed to make an \emph{unbounded} number of edge insertions within communities and edge deletions across communities.
At an intuitive level, this is supposed to only make the problem easier and indeed does so for the minimum bisection approach, but to the contrary \cite{MPW16} proves that the threshold for recovery \emph{increases}.
Dembo, Montanari \& Sen \cite{DMS17} exactly nailed the size of the minimum bisection in \erdos--\renyi graphs, which are complete noise and have no signal in the form of a planted bisection --- and strikingly, it is actually \emph{smaller} than the size of the planted bisection in the detectable regime!
Thus, it is conceivable that there are bisections completely orthogonal to the planted bisection in a stochastic block model graph that nevertheless have the same size.

The problem of recovering communities is more related to the task of \emph{Bayesian inference}, i.e., applying Bayes' rule and approximating $\bx|\bG$.
Optimizing for the minimum bisection is akin to computing the \emph{maximum likelihood estimate}, which does not necessarily produce samples representative of the posterior distribution of $\bx|\bG$.

\medskip

\noindent {\bf SDPs for inference.}
The work of Banks, Mohanty \& Raghavendra \cite{BMR21} proposed a semidefinite programming-based algorithm for inference tasks that incorporates the prior distribution in the formulation, and illustrated that this algorithm can distinguish between $\bG$ sampled from the stochastic block model from an \erdos--\renyi graph of equal average degree anywhere above the KS threshold while being resilient to $\Omega(n)$ arbitrary edge insertions and deletions.

A similar SDP formulation was later studied by Ding, d'Orsi, Nasser \& Steurer \cite{DdNS22} in the $2$-community setting, and was used to give an algorithm to recover the communities with a constant advantage over random guessing in the presence of $\Omega(n)$ edge corruptions for all degrees above the KS threshold.
They analyze the spectra of matrices associated with random graphs after deleting vertices with large neighborhoods, which introduces unfriendly correlations, and causes their analysis to be highly technical.

The main contribution of our work is an algorithm for robust recovery, which is amenable to a significantly simpler analysis.
Our algorithm also succeeds at the recovery task for arbitrary block models with a constant number of communities.
\begin{theorem}[Informal statement of main theorem]
    Let $(\Model, \pi, d)$ be SBM parameters such that $d$ is above the KS threshold, and let $\bG,\bx\sim\SBM_n(\Model, \pi, d)$. There exists $\delta = \delta(\Model, \pi, d) > 0$ such that the following holds. There is a polynomial time algorithm that takes as input any graph $\wt{\bG}$ that can be obtained by performing arbitrary $\delta n$ edge insertions and deletions to $\bG$ and outputs a coloring $\wh{\bx}$ that has ``constant correlation'' with $\bx$, with high probability over the randomness of $\bG$ and $\bx$.
\end{theorem}

Many of the ingredients in the above result are of independent interest.
First, we exhibit a symmetric matrix closely related to the Bethe Hessian of the graph, such that its bottom eigenspace is correlated with the communities.  
Next, we design an efficient algorithm to robustly recover the bottom-$r$ eigenspace of a {\it sparse} matrix in the presence of adversarial corruptions.
Finally, we demonstrate a general rounding scheme to obtain community assignments from this eigenspace.

\begin{remark}[Robustness against node corruptions]
    The node corruption model, introduced by Liu \& Moitra \cite{LM22}, is a harsher generalization of the edge corruption model.
    In recent work, Ding, d'Orsi, Hua \& Steurer \cite{HDdOS23} proved that in the setting of sparse SBM, any algorithm that is robust to edge corruptions can be turned into one robust to node corruptions in a blackbox manner.
    Hence, our results apply in this harsher setting too.
\end{remark}

\subsection{Related work}

We refer the reader to the survey of Abbe \cite{Abb17} for a detailed treatment of the rich history and literature on community detection in block models, its study in other disciplines, and the many information-theoretic and computational results in various parameter regimes.

Introducing an adversary into the picture provides a beacon towards algorithms that utilize but do not \emph{overfit} to distributional assumptions.
Over the years, a variety of adversarial models have been considered, some of which we survey below.

\medskip

\noindent {\bf Corruption models for stochastic block model.}
Prior to the works of \cite{BMR21,DdNS22}, Stefan \& Massouli{\'e} \cite{SM19} considered the robust recovery problem, and gave a robust spectral algorithm to recover communities under $O(n^{\eps})$ adversarial edge corruptions for some small enough $\eps > 0$.

Liu \& Moitra \cite{LM22} introduced the \emph{node corruption} model where an adversary gets to perform arbitrary edge corruptions incident to a constant fraction of corrupted vertices, and gave algorithms that achieved optimal accuracy in the presence of node corruptions and the monotone adversary sufficiently above the KS threshold.
Soon after, Ding, d'Orsi, Hua \& Steurer \cite{HDdOS23} gave algorithms achieving the Kesten--Stigum threshold using algorithms for the edge corruption model in the low-degree setting \cite{DdNS22}, and results on the optimization SDP in the high-degree setting \cite{MS16} in a blackbox manner.

\medskip

\noindent {\bf Semirandom \& smoothed models.}
Some works have considered algorithm design under harsher adversarial models, where an adversarially chosen input undergoes some random perturbations.

Remarkably, at this point, the best algorithms for several graph and hypergraph problems match the performance of our best algorithms for their completely random counterparts.
For example, at this point, the semirandom planted coloring and clique problems were introduced by Blum \& Spencer \cite{BS95}, and Feige \& Kilian \cite{FK01}, and a line of work \cite{CSV17,MMT20} culminating in the work of Buhai, Kothari \& Steurer \cite{BKS23} showed that the size of the planted clique/coloring recoverable in the semirandom setting matches the famed $\sqrt{n}$ in the fully random setting.

Another example where algorithms for a semirandom version of a block model-like problem have been considered is semirandom CSPs with planted solutions, where the work of Guruswami, Hsieh, Kothari \& Manohar \cite{GHKM23} gives algorithms matching the guarantees of solving fully random planted CSPs.

\subsection{Organization}
In \Cref{sec:tech-overview}, we give an overview of our algorithm and proof.
In \Cref{sec:prelim}, we give some technical preliminaries.
In \Cref{sec:recovery-algo}, we describe our algorithm and show how to analyze it.
In \Cref{sec:outlier-eigs}, we prove results about the spectrum of the Bethe Hessian matrix and variants, key to our algorithm.
In \Cref{sec:robust-eigenspace}, we show how to use a trimming procedure and a spectral algorithm on a Bethe Hessian variant to recover a subspace which has good correlation with the true communities in the presence of adversarial corruptions.
Finally, in \Cref{sec:rounding}, we prove some technical claims relevant to analyzing our algorithm.

\section{Technical overview}    \label{sec:tech-overview}
An $n$-vertex graph $\bG$ is drawn from a stochastic block model and undergoes $\delta n$ adversarial edge corruptions, and then the corrupted graph $\wt{\bG}$ is given to us as input.
For simplicity of discussion, we restrict our attention to \emph{assortative} symmetric $k$-community block models above the KS threshold, i.e. the connection probability between two vertices $i$ and $j$ only depends on whether they belong to the same community or different communities, and the intra-community probability is higher.
Nevertheless, our approach generalizes to any arbitrarily specified $k$-community block model above the KS threshold.

Let us first informally outline the algorithm; see \cref{sec:recovery-algo} for formal details.
\begin{enumerate}
    \item First, we preprocess the corrupted graph $\wt{\bG}$ by truncating high degree vertices, which removes corruptions localized on small sets of vertices in the graph.
    \item We then construct an appropriately defined graph-aware symmetric matrix $M_{\bG} \in \R^{n \times n}$ whose negative eigenvalues contains information about the true communities for the \emph{uncorrupted} graph. We motivate this construction in \cref{sec:spectral-overview}.
    \item We recursively trim the rows and columns of $M_{\wt{\bG}}$ to remove small negative eigenvalues in its spectrum. Then we use a spectral algorithm to robustly recover a subspace $U$ which contains information about the communities.  Both points are described in \cref{sec:robust-eigenspace-overview}.
    \item Finally, we round the subspace $U$ into a community assignment, using a vertex embedding provided by $U$. This is detailed in \cref{sec:rounding-overview}.
\end{enumerate}

\subsection{Outlier eigenvectors for the Bethe Hessian}
Bordenave, Lelarge \& Massouli{\'e} \cite{BLM15} analyzed the spectrum of the \emph{nonbacktracking matrix} and rigorously established its connection to community detection.
The asymmetric nonbacktracking matrix $B_G \in \qty{0, 1}^{2\abs{E(G)} \times 2\abs{E(G)}}$ is indexed by directed edges, with 
\begin{align*}
    (B_G)_{(u_1 \to v_1), (u_2 \to v_2)} \triangleq \Ind[v_1 = u_2]\Ind[v_2 \neq u_1].
\end{align*}
\cite{BLM15} showed that above the KS threshold, the $k$ outlier eigenvalues for $B_{\bG}$ correspond to the $k$ community vectors.
More precisely, in the case of symmetric $k$-community stochastic block models above the KS threshold, \cite{BLM15} proved that for the randomly drawn graph $\bG$, there is a small $\eps > 0$ for which its nonbacktracking matrix $B_{\bG}$ has exactly $k$ eigenvalues larger than $(1+\eps){\sqrt{d}}$ in magnitude (\Cref{thm:BLM}). 

The \emph{Bethe Hessian} matrix is a symmetric matrix associated with a graph, whose early appearances can be traced to the works of Ihara \cite{Iha66} and Bass \cite{Bas92}. The Bethe Hessian of a graph with parameter $t \in \R$ is defined as
\begin{align*}
    H_G(t) \triangleq (D_G-I)t^2 - A_Gt + I,
\end{align*}
where $D_G$ and $A_G$ are the diagonal degree matrix and adjacency matrix of $G$, respectively. For $t$ in the interval $[0, 1]$, it can be interpreted as a regularized version of the standard graph Laplacian. 
The Bethe Hessian for the stochastic block model was considered in the empirical works \cite{KMMNSZZ13,SKZ14}, where they observed that for some choice of $t$, the Bethe Hessian and the nonbacktracking matrix has outlier eigenvectors which can be used for finding communities in block models.
Concretely, in \cite{SKZ14} they observed that for $\bG$ drawn from stochastic block models above the KS threshold, there is a choice of $t$ such that $H_{\bG}(t)$ only has a small number of negative eigenvectors, all of which correlate with the hidden community assignment.

We confirm this empirical observation in the following proposition.

\begin{proposition}[Bethe Hessian spectrum]\label{prop:bethe-hessian}
    Let $(\Model, \pi, d)$ be $k$-community SBM parameters such that $d$ is above the KS threshold, and let $\bG,\bx\sim\SBM_n(\Model, \pi, d)$.
    Then there exists $\eps > 0$ such that for $t^* = \frac{1}{(1+\eps)\sqrt{d}}$, the Bethe Hessian $H_{\bG}(t^*)$ has at most $k$ negative eigenvalues and at least $k-1$ negative eigenvalues.
\end{proposition}


\paragraph{Constructing the outlier eigenspace} \label{sec:spectral-overview}
 
There are two assertions in \cref{prop:bethe-hessian}.  To show that $H_{\bG}(t^*)$ has at most $k$ negative eigenvalues, one can relate these negative eigenvalues to the $k$ outlier eigenvalues of $B_{\bG}$ using an Ihara--Bass argument and use a continuity argument as outlined in Fan and Montanari \cite[Theorem 5.1]{FM17}.
 
The more interesting direction is to exhibit at least $k-1$ negative eigenvalues; we will explicitly construct a $k-1$ dimensional subspace starting with the community vectors to witness the negative eigenvalues for $H_G(t^*)$.

Let $\Ind_c$ denote the indicator vector for the vertices belonging to community $c$ and $\Ind$ the all-ones vector. We show that every vector in the span of $\{A^{(\ell)}(\Ind_c - \frac{1}{k} \Ind)\}_{c\in[k]}$ achieves a negative quadratic form against $H_{\bG}(t^*)$, where $A^{(\ell)}$ is the $n\times n$ matrix where the $(i,j)$-th entry encodes the number of length-$\ell$ nonbacktracking walks between $i$ and $j$. This demonstrates a $(k-1)$-dimensional subspace on which the quadratic form is negative.  Formally, we show the following:.
\begin{proposition}\label{prop:new-matrix}
    Under the same setting and notations as \cref{prop:bethe-hessian}, for $\ell \ge 0$ define
    \begin{align*}
        M_{\bG,\ell}\triangleq A^{(\ell)} H_{\bG}(t^*) A^{(\ell)}.
    \end{align*}
    For $\ell = \Theta\left(\tfrac{\log(1/\eps)}{\eps}\right)$ and every $c \in [k]$, we have 
    \begin{align*}
        \ev{\Ind_c - \tfrac{1}{k} \Ind, M_{\bG,\ell}(\Ind_c - \tfrac{1}{k} \Ind)} \le -\Omega(n).
    \end{align*}
    Hence,  $M_{\bG,\ell}$ has at most $k$ negative eigenvalues and at least $k-1$ negative eigenvalues.
\end{proposition}
Nonbacktracking powers and related constructions were previously  studied in \cite{Mas14,MNS18}, but there they take $\ell = \Theta(\log n)$, whereas we only consider constant $\ell$. Besides simplifying the analysis of the quadratic form, using constant $\ell$ is also critical for tolerating up to $\Omega(n)$ corruptions. 

As a consequence of \cref{prop:new-matrix}, the negative eigenvectors of $M_{\bG, \ell}$ are correlated with the centered community indicators $\{\Ind_c - \frac{1}{k} \Ind\}_{c \in [k]}$, while the negative eigenvectors of $H_{\bG}(t^*)$ are correlated with $\{A^{(\ell)}(\Ind_c - \frac{1}{k} \Ind)\}_{c\in[k]}$. 
The upshot is that we can directly use the negative eigenvectors of $M_{\bG, \ell}$ to recover the true communities in the absence of corruptions.

\begin{remark}
    Based on the empirical observations in \cite{KMMNSZZ13,SKZ14}, a natural hope is to directly use the Bethe Hessian for recovery. However, it turns out that the quadratic form of the centered true community indicators $\ev{(\Ind_c - \tfrac{1}{k} \Ind), H_{\bG}(t^*) (\Ind_c - \tfrac{1}{k} \Ind)}$ are actually \emph{positive} close to the KS threshold, so the same approach does not establish that the negative eigenvectors of $H_{\bG}(t^*)$ correlate with the communities. 
\end{remark}

We will now discuss how to recover the outlier eigenspace in the presence of adversarial corruptions.

\subsection{Robust PCA for sparse matrices }\label{sec:robust-eigenspace-overview}
The discussion above naturally leads to the following algorithmic problem of robust recovery: Given as input a corrupted version $\wt{M}$ of a symmetric matrix $M$, can we recover the bottom/top $r$-dimensional eigenspace of $M$?
Since the true communities are constantly correlated with the outlier eigenspace of $M = M_{\bG,\ell}$, recovering the outlier eigenspace of $M$ from its corrupted version $\wt{M} = \wt{M}_{\bG,\ell}$ is a major step towards robustly recovering communities.

The problem of robustly recovering the top eigenspace, a.k.a. robust PCA has been extensively studied, and algorithms with provable guarantees have been designed (see \cite{CLMW11}).  
However, the robust PCA problem in our work is distinct from those considered in the literature in a couple of ways.
For us, the uncorrupted matrix $M$ is sparse and both the magnitude and location of the noisy entries are adversarial.
Furthermore, for our purposes, we need not recover the actual outlier eigenspace of $M$.
Indeed, as we discuss below, it suffices to robustly recover a constant dimensional subspace which is constantly correlated with the true communities. 

We design an efficient algorithm to robustly recover such a subspace under a natural set of sufficient conditions on $M$.
Before we describe these conditions, let us fix some notation.
We will call a vector $x \in \R^n$ to be $C$-delocalized if no coordinate is large relative to others, i.e.,  $|x_i|^2 \leq \frac{C}{n} \|x\|^2$ for all $i \in [n]$. 
Delocalization has previously been used in the robust PCA literature under the name ``incoherence'' \cite{CLMW11}. 

   Let $M$ be a $n \times n$ matrix with at most $r$ negative eigenvalues.  In particular, the $r$-dimensional negative eigenspace $V_M$ of $M$ is the object of interest.
    Let $\wt{M}$ be a corrupted version of $M$, differing from $M$ in $\delta n$ coordinates.

    Given the corrupted version $\wt{M}$, a natural goal would be to recover the $r$-dimensional negative eigenspace $V_M$.  
    It is easy to see that it could be impossible to recover the space $V_M$.
    Instead, we will settle for a relaxed goal, namely, recover a slightly larger dimensional subspace $U$ that non-trivially correlates with delocalized vectors in the true eigenspace $V_M$.  
    More formally, we will solve the following problem.
\begin{restatable}{problem}{problemsubspace}\label{prob:subspace-recovery}
    Given the corrupted matrix $\wt{M}$ as input, give an efficient algorithm to output a subspace $U$ with the following properties:
    \begin{enumerate}
        \item {\bf Low dimensional.} The dimension of $U$ is $O\parens*{r}$.
        \item {\bf Delocalized.}
        The diagonal entries of its projection matrix $\Pi_U$ are bounded by $O\parens*{\frac{r}{n}}$.
        \item {\bf Preserves delocalized part of negative eigenspace.} For any $C$-delocalized unit vector $y$ such that $\angles*{y, M y} < -\Omega(1)$, we have $\angles*{y, \Pi_U y} \ge \Omega\parens*{1}$.
    \end{enumerate}
\end{restatable}

In fact, our algorithm will recover a principal submatrix of $\wt{M}$ whose eigenspace $V$ for eigenvalues less than $-\eta$ is $O(r)$-dimensional. 
Moreover, the eigenspace $V$ can be processed to another delocalized, $O(r)$-dimensional subspace $U$ that satisfies the conditions outlined above.

Although the matrix $M$ has a constant number of negative eigenvalues, its corruption $\wt{M}$ can have up to $\Omega(n)$ many. 
At first glance, it may be unclear how a constant dimensional subspace $U$ can be extracted from $\wt{M}$. 
The crucial observation is that the large negative eigenvalues introduced by the corruptions are highly localized.
Thus, we will design an iterative trimming algorithm that aims to delete rows and columns to clean up these localized corruptions.
When the algorithm terminates, it yields the $O(r)$-dimensional subspace $V$.

\medskip

\noindent {\bf Recovering a principal submatrix.}
We now describe the trimming algorithm informally and refer the reader to \cref{sec:robust-eigenspace} for the formal details.

We first fix some small parameter $\eta > 0$ and execute the following procedure, which produces a series of principal submatrices $\wt{M}^{(t)}$ for $t \ge 0$, starting with $\wt{M}^{(0)} \triangleq \wt{M}$. 
\begin{enumerate}
    \item At step $t$, if the eigenspace $V$ of eigenvalues of $\wt{M}^{(t)}$ less than $-\eta$ is $O(r)$-dimensional, we terminate the algorithm and output $V$.
    \item Otherwise, compute the projection $\Pi^{(t)}$ corresponding to the $\le -\eta$ eigenspace of $\wt{M}^{(t)}$. 
    \item Sample an index $i \in [n]$ of $\wt{M}^{(t)}$ with probability proportional to $\Pi^{(t)}_{i, i}$. 
    \item Zero out row and column $i$, and set this new principal submatrix as $\wt{M}^{(t+1)}$.
\end{enumerate}

We now discuss the intuition behind the procedure and formally prove its correctness in \cref{subsec:cleaned-up-spectrum}.
The main idea of step 3 is that one should prefer to delete highly localized eigenvectors which have relatively large negative eigenvalues.
This is reasonable because the size of the diagonal entries of $\wt{M}^{(t)}$ serve as a rough proxy for the level of delocalization. 

As a concrete illustration of this intuition, suppose that $\wt{M} = \Pi^{(0)} = -uu^\top - vv^\top$, where $u, v$ are orthogonal unit vectors.
Moreover, suppose $u$ is $C$-delocalized whereas $v = e_1$.
Then $\Pi^{(0)}_{1, 1} = 1$ whereas $|\Pi^{(0)}_{i, i}| \le C^2/n$ for $i > 1$.
Hence, deleting the first row and column of $\wt{M}$ also deletes the localized eigenvector $v$.
In general, whenever one of the eigenvectors of $\wt{M}^{(t)}$ is heavily localized on a subset of coordinates $S$, the diagonal entries in $\Pi^{(t)}_{S, S}$ are disproportionately large. 
This leads to a win-win scenario: either we reach the termination condition, or we are likely to mitigate the troublesome large localized eigenvectors.

We now discuss how we achieve the second and third guarantees in \cref{prob:subspace-recovery}.

\medskip
\noindent{\bf Trimming the subspace.}
The final postprocessing step is simple.
Let $V$ denote the eigenspace with eigenvalues less than $-\eta$ for the matrix $\wt{M}^{(T)}$ obtained at end of iterative procedure.

To ensure delocalization (condition 2 in \cref{prob:subspace-recovery}), the idea is to take its projector $\Pi_V$ and trim away the rows and columns with diagonal entry exceeding $\frac{\tau}{n}$ for some large parameter $\tau > 0$. 
The desired delocalized subspace $U$ is obtained by taking the eigenspace of the trimmed $\Pi_V$ corresponding to the eigenvalues exceeding a threshold that is $O(\eta)$. 
Since $V$ is $O(r)$-dimensional, so too is $U$.

The more delicate part is condition 3 in \cref{prob:subspace-recovery}. Namely, we must show that despite corruptions and the repeated trimming steps, $x$ remains a delocalized witness vector for $\Pi_U$, and thus has constant correlation with the subspace $U$.
The key intuition for this is that delocalized witnesses are naturally robust to adversarial corruptions, so long as the adversarial corruptions have bounded row-wise $\ell_1$ norm. 
In particular, since delocalization is an $\ell_{\infty}$ constraint, H\"{o}lder's inequality bounds the difference in value of the quadratic form using $M$ and $\wt{M}$.
In \cref{subsec:postproc-props}, we prove that for sufficiently small constant levels of corruption, $x$ is also a delocalized witness for $\wt{M}$ and $\Pi_U$.

Finally, we discuss how to round the recovered subspace $U$ into a community assignment.

\subsection{Rounding to communities}    \label{sec:rounding-overview}
At this stage, we are presented with a constant-dimensional subspace $U$ with the key feature that it is correlated with the community assignment vectors $\{\Ind_c\}_{c\in[k]}$.
Our goal is to round $U$ to a community assignment that is ``well-correlated'' with the ground truth.
In order to discuss how we achieve this goal, we must make precise what it means to be ``well-correlated'' with the ground truth.
Notice that a community assignment is just as plausible as the same assignment with the names of communities permuted, and thus counting the number of correctly labeled vertices is not a meaningful metric.

A more meaningful metric is the number of pairwise mistakes, i.e.~the number of pairs of vertices in the same community assigned to different communities or in different communities assigned to the same community.
A convenient way to express this metric is via the inner product of positive semidefinite matrices encoding whether pairs of vertices belong to the same community or not.
Given a community assignment $x$, we assign it the matrix $X$, defined as
\[
    X[i,j] =
    \begin{cases}
        1 &\text{if $x(i) = x(j)$} \\
        -\frac{1}{k-1} & \text{if $x(i)\ne x(j)$}.
    \end{cases}
\]
For the ground truth assignment $\bx$ and the output of our algorithm $\wh{x}$, we measure the correlation with $\langle\bX,\wh{X}\rangle$.
Observe that for any guess $\wh{X}$ that is oblivious to the input (for example, classifying all vertices to the same community, or blindly guessing), the value of $\langle\bX,\wh{X}\rangle$ is concentrated below $\Tilde{O}(n^{3/2})$.
On the other hand, if $\wh{X} = \bX$, then this correlation is $\Omega(n^2)$.
See \Cref{def:weak-recovery,sec:justify-metric} for how this notion generalizes to arbitrary block models, and subsumes other notions of weak-recovery defined in literature.

The projection matrix $\Pi_U$ satisfies $\angles*{\Pi_U, \bX}\ge\Omega\parens*{\norm*{\Pi_U}_F \cdot \norm*{\bX}_F } = \Omega(n)$.
We give a randomized rounding strategy according to which $\E\wh{X} \psdge c \cdot n \cdot \Pi_U$ for some constant $c > 0$.
Consequently, $\E\langle\bX,\wh{X}\rangle = cn\cdot\angles*{\Pi_U,\bX} \ge \Omega(n^2)$.

Observe that for any community assignment $x$, its matrix representation $X$ is rank-$(k-1)$, which lets us write it as $VV^{\top}$ for some $n\times (k-1)$ matrix $V$.
Here, the $i$-th row of $V$ is some vector $v_{x(i)}$ that only depends only on the community $x(i)$ where vertex $i$ is assigned.

Our rounding scheme uses $\Pi_U$ to produce an embedding of the $n$ vertices as rows of a $n\times(k-1)$ matrix $W$ whose rows are in $\{v_1,\dots,v_k\}$. 
In the community assignment $\wh{x}$ outputted by the algorithm, the $i$-th vertex is assigned to community $j$ if the $i$-th row of $W$ is equal to $v_j$.
We then show that $\E WW^{\top} \psdge c\cdot n\cdot \Pi_U$.
Since $\wh{X} = \E WW^{\top}$, we can conclude $\E\langle\bX,\wh{X}\rangle \ge \Omega(n^2)$.

\medskip
\noindent {\bf Rounding scheme.}
Our first step is to obtain an embedding of the $n$ vertices into $\R^{k-1}$ by choosing a $(k-1)$-dimensional random subspace $U'$ of $U$, then writing its projector as $M' M'^{\top}$, and choosing the embedding as the rows of $M'$: $u_1',\dots,u_n'$.
Suppose this embedding has the property that for some $c' > 0$, the rows of $\sqrt{c'n}U'$ lie inside the convex hull of $v_1,\dots,v_k$, then we can express each $u_i'$ as a convex combination $\sum_{j=1}^k p^{(i)}_{j} v_j$ and then independently sample $w_i$ from $\{v_1,\dots,v_k\}$ according to the probability distribution $(p^{(i)}_j)_{j\in[k]}$.
The resulting embedding $W$ would satisfy the property that $\E WW^{\top} \psdge c'\cdot \frac{k-1}{\mathrm{dim}(U)}\cdot n\cdot \Pi_U$, where this inequality holds since the off-diagonal entries are equal, and the diagonal of $WW^{\top}$ is larger.

The reason an appropriate scaling $c'$ exists follows from the facts that the convex hull of $v_1,\dots,v_k$ is full-dimensional and contains the origin, which we prove in \Cref{sec:rounding}.

%% file: prelim.tex
\section{Preliminaries} \label{sec:prelim}

\noindent {\bf Stochastic block model notation.}
We write $\bone$ to denote the all-ones vector and $e_i$ to denote the $i$th standard basis vector, with the dimensions implicit. For a $k$-community block model, let $\pi \in \R^k$ denote the prior community probabilities, and $\Pi = \diag(\pi)$, so that $\pi = \Pi \bone$. The edge probabilities are parameterized by a symmetric matrix $\Model \in \R^{k \times k}$, the block probability matrix. A true community assignment $\bx : [n] \to [k]$ is sampled i.i.d. from $\pi$. Conditioned on $\bx$, an edge between $i$ and $j$ is sampled with probability $\frac{\Model_{\bx(i), \bx(j)}d}{n}$. To ensure that the average degree is $d$, we stipulate that $\Model\pi = \bone$. 

We will also use $\bX\in \R^{n\times k}$ to denote the \emph{one-hot encoding} of $\bx$, i.e., the matrix where the $t$-th row is equal to $e_{\bx(t)}$. We will sometimes find it convenient to access the columns of $\bX$, which are the indicator vectors for the $k$ different communities; we denote these by $\bone_c$ for any community $c \in [k]$. For any $f:[k]\to\R$, define the lift of $f$ with respect to the true community assignment by $\boldf^{(n)}\triangleq \sum_{c\in[k]} f(c)\cdot\Ind_c$. 

Another natural matrix that appears throughout the analysis is the Markov transition matrix $T \triangleq M\Pi$, which by detailed balance evidently has stationary distribution $\pi$. This is an asymmetric matrix, but since $T$ defines a time-reversible Markov chain with respect to $\pi$, $T$ is self-adjoint with respect to the inner product $\ev{\cdot, \cdot}_{\pi}$ in $\R^k$ induced by $\pi$. Hence $T$ is diagonalizable with real eigenvalues and its eigenvalues are $1 = \lambda_1 > \abs{\lambda_2} \ge \cdots \ge \abs{\lambda_{k}}$, with ties broken by placing positive eigenvalues before the negative ones.
Note that the normalization condition $\Model\pi = \bone$ translates into $T\bone = \bone$.

\medskip

\noindent{\bf Matrix notation.}
We use $\psdle$ and $\psdge$ to denote inequalities on matrices in the Loewner order.
For any $n\times n$ matrix $X$, we use $\Pi_{\le a}(X)$ and $\Pi_{\ge a}(X)$ to denote the projectors onto the spaces spanned by eigenvectors of $X$ with eigenvalue at most and at least $a$ respectively. We also define $X_{\le a}\triangleq \Pi_{\le a}(X) X \Pi_{\le a}(X)$ and $X_{\ge a} \triangleq \Pi_{\ge a}(X) X \Pi_{\ge a}(X)$, the corresponding truncations of the eigendecomposition of $X$. 

For $S\subseteq[n]$, we use $X_{S,S}$ to denote the matrix obtained by taking $X$ and zeroing out all rows and columns with indices outside $S$.

\medskip

\noindent {\bf Nonbacktracking matrix and Bethe Hessian.} For a graph $G$, let $B_G$ be its nonbacktracking matrix, $A_G$ be its adjacency matrix, $D_G$ be its diagonal matrix of degrees, $A_G^{(\ell)}$ be its $\ell$-th nonbacktracking power of $A_G$, and $H_G(t)\triangleq (D_G-I)t^2 - A_Gt + I$ be its Bethe Hessian matrix.
The matrix we use for our algorithm is $M_{G,\ell}(t)\triangleq A_G^{(\ell)} H_G(t) A_G^{(\ell)}$.
We will drop the $G$ from the subscript when the graph $G$ is clear from context.

\medskip

\noindent {\bf Determinants.}
Below, we collect some standard linear algebraic facts that will prove useful.
\begin{fact}    \label{fact:every-submatrix-singular}
    Suppose a matrix $X$ has a kernel of dimension $k$, then every $(n-j)\times(n-j)$ submatrix of $X$ for $j < k$ is singular.
\end{fact}

\begin{fact}[Jacobi's formula]
    For any differentiable function $X:\R\to\R^{n\times n}$,
    \[
        \frac{d}{du}\det(X(u)) = \sum_{i=1}^n \sum_{j=1}^n \det\parens*{X(u)_{[n]\setminus\{i\}, [n]\setminus\{j\}}} \cdot \parens*{-1}^{i+j} \cdot \frac{d}{du} (X(u))_{i,j}.
    \]
\end{fact}
\begin{lemma}   \label{lem:jacobi-iterate}
    Let $X:\R\to\R^{n\times n}$ and $f:\R\to\R$ be any pair of smooth functions.
    For any $j\ge 0$, there exist functions $(q_{S,T}:\R\to\R)_{S,T\subseteq[n],~|S|=|T|\ge n-j}$ such that:
    \[
        \parens*{\frac{d}{du}}^j \bracks*{\det(X(u))\cdot f(u)} = \sum_{S,T\subseteq[n],~|S|=|T|\ge n-j} \det\parens*{X(u)_{S,T}} q_{S,T}(u).
    \]
\end{lemma}
\begin{proof}
    We prove this by induction.
    This is clearly true when $j = 0$, and the induction step is a consequence of Jacobi's formula.
\end{proof}

\medskip

\noindent {\bf Kesten-Stigum threshold.}
We say that a stochastic block model is above the Kesten--Stigum (KS) threshold if $\lambda_2(T)^2 d > 1$, where recall that $\lambda_2$ is the second largest eigenvalue in absolute value.
We use $r$ to denote the number of eigenvalues of $T$ equal to $\lambda_2(T)$.

%% file: recovery-algo.tex
\section{Recovery algorithm}\label{sec:recovery-algo}
Let $\bG$ be the graph drawn from $\SBM_n(\Model,\pi,d)$, and let $\wt{\bG}$ denote the input graph which is $\bG$ along with an arbitrary $\delta n$ adversarial edge corruptions.
Our algorithm for clustering the vertices into communities proceeds in multiple phases, described formally below.

The first phase preprocesses the graph by making it bounded degree and constructs an appropriate matrix $M$ associated to the graph.
The second phase cleans up $M$ and uses a spectral algorithm to robustly recover a subspace containing nontrivial information about the true communities.
Finally, the third phase rounds the subspace to an actual community assignment.

\noindent\rule{16cm}{0.4pt}
\begin{algorithm}
$\wt{\bG}$ is given as input, and a community assignment to the vertices is produced as output.

\noindent {\bf Phase 1: Deletion of high-degree vertices.}
For some large constant $B > 0$ to be specified later, we perform the following truncation step: delete all edges incident on vertices with degree larger than $B$ in $\wt{\bG}$. This forms a graph $\wt{\bG}_B$, with corresponding adjacency matrix $A_{{\wt{\bG}}_B} \in \R^{\abs{V(\bG)} \times \abs{V(\bG)}}$. To avoid confusion, we preserve the vertex set $V(\bG)$, but it should be understood that the truncated vertices do not contribute to the graph.

For technical considerations, we also define a (nonstandard) truncated diagonal matrix 
\begin{equation}
\ol{D}_{{\wt{\bG}}_B} \triangleq \diag\left(\deg(v)\bone[\deg(v) \le B]\right)_{v \in V(\bG)}
\end{equation} 
With this, we can then define the truncated Bethe Hessian matrix 
\begin{equation}
\ol{H}_{\wt{\bG}_B}(t) \triangleq I - tA_{\wt{\bG}_B} + t^2(\ol{D}_{\wt{\bG}_B} - I).
\end{equation} 
Finally, the input matrix to the next phase is 
\begin{equation}
\ol{M}_{\wt{\bG}_B, \ell}(t) \triangleq A_{\wt{\bG}_B}^{(\ell)} \ol{H}_{\wt{\bG}_B}(t) A_{\wt{\bG}_B}^{(\ell)},
\end{equation}
where we also choose the value of $t$ later.

\begin{remark}
To reduce any chance of confusion with the notation, we reiterate our conventions for distinguishing between different versions of various matrices. If a graph is truncated at level $B$, then we add a subscript $B$. We use tilde to denote that we are working with a corrupted graph. Finally, we use overline to denote that we are working with the nonstandard version of the Bethe Hessian after truncation.  

For example, the matrix $\ol{D}_{\bG_B}$ no longer corresponds to the degree matrix of $\bG_B$, since as stated it still counts edges from truncated vertices. This is done to simplify the analysis of the spectrum of $\ol{M}_{\bG_B, \ell}(t)$ but we do not believe it to be essential. 
\end{remark}

\medskip

\noindent {\bf Phase 2: Recovering a subspace with planted signal.}
Define $M\triangleq \ol{M}_{\wt{\bG}_B,\ell}$.
We give an iterative procedure to ``clean up'' $M$ by deleting a few rows and columns.
We then run a spectral algorithm on the cleaned up version of $M$.

Let $\eta > 0$ be a small constant we choose later, and let $K \triangleq B^{2\ell+3}$.
\begin{enumerate}
    \item Define $M^{(0)}$ as $M$.
    Let $t$ be a counter initialized at $0$, and $\Phi(X)$ as the number of eigenvalues of $X$ smaller than $-\eta$.
    \item While $\Phi(M^{(t)}) > \frac{2K}{\eta}r$: compute the projection matrix $\Pi^{(t)}\triangleq \Pi_{\le -\eta}(M^{(t)})$, choose a random $i\in[n]$ with probability $\frac{\Pi^{(t)}_{i,i}}{\Tr\parens*{\Pi^{(t)}}}$, and define $M^{(t+1)}$ as the matrix obtained by zeroing out the $i$-th row and column of $M^{(t)}$.
    Then increment $t$.
\end{enumerate}
Let $T$ be the time of termination and $\tau > 0$ be a large enough constant we choose later.
We compute $\Pi^{(T)}$, and then compute as the set $S$ of all indices $i$ where $\Pi^{(T)}_{i,i} \le \frac{\tau}{n}$.
Define $\wt{\Pi}$ as $\parens*{\Pi^{(T)}_{S,S}}_{\ge \eta/K}$, and compute its span $U$, where we recall that $\parens*{X}_{\ge a}$ denotes the truncation of the eigendecomposition of $X$ for eigenvalues at least $a$.
This subspace $U$ is passed to the next phase.

\medskip

\noindent {\bf Phase 3: Rounding to a community assignment.}
Define $r'$ as $r-1$ when $\lambda_2(T) > 0$ and as $r$ when $\lambda_2(T) < 0$.
We first obtain an $r'$-dimensional embedding of the vertices into $\R^{r'}$.
Compute a random $r'$-dimensional subspace $U'$ of $U$, and take an orthogonal basis $u_1',\dots,u'_{r'}$.
Place these vectors as a column of a matrix $M'$ in $\R^{n\times r'}$.
The rows of $M'$ gives us the desired embedding.

On the other hand, we use the natural embedding of the $k$ communities into $\R^{r'}$ induced by the $r'$ nontrivial right eigenvectors corresponding to the eigenvalue $\lambda_2(T)$: $(\psi_i)_{1\le i\le r'}$ of $T$.
In particular, let $\Psi_{r'} \triangleq \mqty[\psi_1 & \cdots & \psi_{r'}] \in \R^{k \times r'}$ be the matrix of these $r'$ nontrivial eigenvectors of $T$.
Then the row vectors $\phi_1, \ldots, \phi_k \in \R^{r'}$ form the desired embedding of communities.

In the rounding algorithm, we first find the largest $c$ such that all the rows of $c\cdot M'$ lie in the convex hull of $\phi_1, \ldots, \phi_k$. 
We can find such a value of $c$ if it exists by solving a linear program, and this $c > 0$ is guaranteed to exist by \Cref{lemma:convex-hull}.
Then, for each $i \in [n]$ we express each row of $c \cdot M'$ as a convex combination $\sum_{j=1}^k w_{i}^{(j)} \phi_j$ for nonnegative $w_i^{(j)}$ such that $\sum_{j=1}^k w_i^{(j)} = 1$.
Finally, we assign vertex $i$ to community $j$ with probability $w_i^{(j)}$, and output the resulting community assignment $\wh{\bx}$.
\noindent\rule{16cm}{0.4pt}
\end{algorithm}

\begin{remark}
    Scaling the rows of $M'$ so as to lie in the convex hull of $\{\phi_j\}_{j \in [k]}$, is reminiscent of the rounding algorithm of Charikar \& Wirth \cite{CW04} to find a cut of size $\frac{1}{2} + \Omega(\frac{\eps}{\log(1/\eps)})$ in a graph with maximum cut of size $\frac{1}{2} + \eps$: in their algorithm, they scale $n$ scalars to lie in the interval $[-1,1]$.
\end{remark}

\subsection{Analysis of algorithm}\label{sec:algo-analysis}
Our goal is to prove that the output $\wh{\bx}$ of our algorithm is well-correlated with the true community assignment $\bx$.
We begin by defining a notion of weak recovery for $k$-community stochastic block models.
\begin{restatable}[Weak recovery]{definition}{weakrecovery}\label{def:weak-recovery}
    Let $\Psi \triangleq \mqty[\psi_2 & \cdots & \psi_{k}] \in \R^{k \times (k-1)}$ be the matrix of the top-$(k-1)$ nontrivial eigenvectors of the transition matrix $T$ of a stochastic block model.
    
    For $\rho > 0$, we say that a (randomized) algorithm for producing community assignments $\wh{\bX} \in \R^{n \times k}$ achieves $\rho$-weak recovery if 
    \begin{align*}
        \ev{\E \wh{\bX}_{\Psi} , \bX_{\Psi}} \ge \rho \norm*{ \E \wh{\bX}_{\Psi}}_F \norm*{\bX_{\Psi}}_F,
    \end{align*}
    where $B_{\Psi} \triangleq (B \Psi) (B \Psi)^\top$ for a matrix $B \in \R^{n \times k}$.
\end{restatable}
\begin{remark}
    Intuitively, this notion is capturing the ``advantage'' of the algorithm over random guessing, or simply outputting the most likely community. 
    See \cref{sec:justify-metric} for a more detailed discussion of this notion, how it recovers other previously considered measures of correlation in the case of the symmetric block model, and why it is meaningful.
    In particular, it implies the notion of weak recovery used in \cite{DdNS22}.
\end{remark}

Our main guarantee is stated below.
\begin{theorem} \label{thm:formal-thm}
    For any SBM parameters $(\Model,\pi,d)$ above the KS threshold, there is a constant $\rho(\Model,\pi,d) > 0$ such that the above algorithm takes in the corrupted graph $\wt{\bG}$ as input and outputs $\wh{\bx}$ achieving $\rho(\Model,\pi,d)$-weak recovery with probability $1-o_n(1)$ over the randomness of $\bG\sim\SBM_n(\Model,\pi,d)$.
\end{theorem}

To prove the above theorem it suffices to analyze $\angles*{(\E \wh{\bX})_{\Psi}, \bX_{\Psi}}$. To see why, let us first set up some notation. For each vertex $i$, we obtain a simplex vector $w_i \in \R^{k}$, which we can stack as rows into a weight matrix $W \in \R^{n \times k}$. We then independently round each vertex so that $\E \wh{\bX} = W$. 

To analyze our rounding scheme, first note that $\E[\wh{\bX}_{\Psi}]$ is equal to $W_{\Psi}$ off of the diagonal and is larger than $W_{\Psi}$ on the diagonal, and thus $\E[\wh{\bX}_{\Psi}]\psdge W_{\Psi}$.
Since $\bX_{\Psi}$ is positive semidefinite, $\angles*{\E[\wh{\bX}_{\Psi}], \bX_{\Psi}} \ge \angles*{W_{\Psi}, \bX_{\Psi}}$.
Thus, it suffices to lower bound $\angles*{W_{\Psi}, \bX_{\Psi}}$. By construction,
$W_{\Psi}$ is equal to $c^2 \cdot \Pi_{U'}$, where recall that $U'$ was a random $r'$-dimensional subspace of $U$, the output of Phase 2 of the algorithm.
Thus,
\[
    \E_{U'} W_{\Psi} = c^2 \cdot \E_{U'} \Pi_{U'} = c^2 \cdot\frac{r'}{\dim(U)} \Pi_U.
\]

In \Cref{lem:subspace-output}, we prove that when $(\Model,\pi,d)$ are above the KS threshold, $\angles*{\Pi_U, \bX_{\Psi}} \ge \Omega(1)\cdot \norm*{\Pi_U}_F \cdot \norm*{\bX_{\Psi}}_F$ and $\dim(U) = O(1)$.
In \Cref{lemma:convex-hull}, we show that when $\diag(\Pi_U) = O(1/n)$, we can take $c = \Omega(\sqrt{n})$; this delocalization condition is guaranteed by phase 2 of the algorithm.
Combined with the fact that $\norm*{\wh{\bX}_{\Psi}}_F = O(n)$, it follows that $\angles*{\E \wh{\bX}_{\Psi}, \bX_{\Psi}} \ge \Omega(1)\cdot\norm*{\wh{\bX}_{\Psi}}_F\cdot\norm*{\bX_{\Psi}}_F$, which establishes \Cref{thm:formal-thm}.

%% file: outlier-bound.tex
\section{The number of outlier eigenvalues} \label{sec:outlier-eigs}
In this section, we pinpoint the number of outlier eigenvalues of $M_{\bG, \ell}(t) \triangleq A_{\bG}^{(\ell)} H_{\bG}(t) A_{\bG}^{(\ell)}$. 
Specifically, we will show the following theorem.
\begin{theorem} \label{thm:main-outlier}
    Let $(\mathrm{M},\pi,d)$ be a model and let $T$ be its corresponding Markov transition matrix.
    Let $r$ be multiplicity of $\lambda_2(T)$.
    There exists constants $\ell,\, \delta$ and $\alpha$, such that for $t = \frac{1+\delta}{\lambda_2{d}}$, with high probability over $\bG\sim\SBM_n(\mathrm{M},\pi,d)$, $M_{\bG, \ell}(t)$ has at least $r$ negative eigenvalues of magnitude at most $-\alpha$, and at most $r+1$ negative eigenvalues.
\end{theorem}

We first establish the upper bound on the number of outlier eigenvalues, and then the lower bound.
To establish the theorem, we will harness the Ihara-Bass formula that connects $H_G(t)$ to the non-backtracking walk matrix, and the work of Bordenave et al.~\cite{BLM15} on the spectrum of non-backtracking walk matrix $B$ in stochastic block models.  We begin by recalling these results from the literature.
\begin{theorem}[Ihara--Bass formula]\label{thm:ihara-bass}
    For any simple, undirected, unweighted graph $G$ on $n$ vertices and $m$ edges, we have 
    \begin{align*}
        \det(I - tB_G) = \det(H_G(t))(1-t^2)^{m-n}
    \end{align*}
\end{theorem}
For a matrix $X$, we use $\lambda_i(X)$ to denote its $i$-th largest eigenvalue when sorted by absolute value.

\begin{theorem}[Nonbacktracking spectrum of SBM \cite{BLM15}]   \label{thm:BLM}
    For $\bG \sim \SBM_n(\mathrm{M}, \pi, d)$:
    \[
        \lambda_i(B_{\bG}) = \lambda_i(T)\cdot d \pm o_n(1)
    \]
    for all $1\le i\le k$ such that $\abs{\lambda_i(T)} > \frac{1}{\sqrt{d}}$ with high probability.
    All the remaining eigenvalues of $B_{\bG}$ are at most $\sqrt{d}\cdot(1+o_n(1))$ in magnitude with high probability.
\end{theorem}

The following is a consequence of the Ihara--Bass formula.
Along with \Cref{thm:BLM}, this implies an upper bound on the number of outlier eigenvalues in \Cref{thm:main-outlier}.
\begin{lemma}   \label{lem:B-to-H}
    For any $t^* > 0$, the number of negative eigenvalues of $H_G(t^*)$ is at most the number of real eigenvalues of $B_G$ larger than $1/t^*$, and for $t^* < 0$, the number of negative eigenvalues of $H_G(t^*)$ is at most the number of real eigenvalues of $B_G$ smaller than $1/t^*$.
\end{lemma}
\begin{proof}
    We present the proof assuming $t^* > 0$, since the case when $t^* < 0$ follows an identical proof.
    First, observe that the number of negative eigenvalues of $H_G(t^*)$ is at most:
    \[
        \sum_{0\le t < t^*:H_G(t)\text{ is singular}} \mathrm{dim}\,\mathrm{kernel}(H_G(t)),
    \]
    since $H_G(0)$ has no negative eigenvalues, and the eigenvalues are a continuous function of $t$.

    By the Ihara--Bass formula, if $H_G(t)$ is singular, then $\det(I-tB_G)$ must be $0$, and in particular, $1/t$ must be an eigenvalue of $B_G$.
    We now prove that $\mathrm{dim}\,\mathrm{kernel}(H_G(t))$ is at most the multiplicity of the eigenvalue $1/t$, from which the desired statement follows.

    Let $r$ be the dimension of the kernel of $H_G(t)$.
    We will show $\qty(\frac{d}{du})^{j}\det(I-uB_G) \Bigr|_{u=t} = 0$ for $0\le j\le r-1$, thereby implying that $t$ occurs as a root of $\det(I-uB_G)$, and hence $1/t$ as an eigenvalue, with multiplicity at least $r$.
    Now, by the Ihara--Bass formula,
    \begin{align*}
        \parens*{\frac{d}{du}}^j\det(I-uB_G) &= \parens*{\frac{d}{du}}^j \bracks*{\det(H(u))\cdot(1-u^2)^{m-n}},
    \end{align*}
    which by \Cref{lem:jacobi-iterate} is equal to
    \[
        \sum_{S,T\subseteq[n],\,|S|=|T|\ge n - j} \det(H(u)_{S,T}) \cdot q_{S,T}(u)
    \]
    for some collection of scalar functions $q_{S,T}$.
    Since $H(t)$ has a kernel of dimension $r$, by \Cref{fact:every-submatrix-singular} $H(t)_{S,T}$ is singular for any choice of $S,T\subseteq[n]$ with $|S|=|T|\ge n-r+1$, and hence the above quantity is $0$ for all $0\le j\le r-1$, which completes the proof.
\end{proof}

%% file: local-stat.tex
\subsection{Lower bound on the number of outlier eigenvalues} \label{sec:local-stat}
In this section, we lower bound the number of negative eigenvalues of $M_{\bG, \ell}(t)$, where $t > 0$. The same arguments apply to the case where $t < 0$, so we suppress the subscript in the dimension $r$ of the outlier eigenspace for ease of notation. 

To lower bound the number of outliers, we will show that with high probability the quadratic form of $M_{\bG, \ell}(t)$ is negative on an $r$-dimensional subspace spanned by vectors corresponding to the true community assignment. To do so, we will calculate the first and second moments of this random quadratic form, but the bulk of the work comes in calculating the first moment. The key insight here is that since these vectors are constructed using the true communities, the expectations can be tractably computed. 

It is fruitful to give a combinatorial interpretation of the matrix $M_{\bG, \ell}(t)$. As is standard in trace moment calculations, the quadratic form can be decomposed into sums over various different graph shapes. Some standard accounting for these shapes then leads us to the final answer.
To concretely articulate what we mean, we first define the notion of a \emph{partially labeled shape} and its associated matrix.
\begin{definition}[Shape matrix]
    A \emph{shape} is a graph $H$ along with two distinguished vertices: a \emph{left vertex} $u_L$ and a \emph{right vertex} $u_R$.
    Given an $n$-vertex graph $G$, its shape matrix of $G$, denoted $M_H(G)$, is an $n\times n$ matrix where the $(i,j)$ entry is the number of copies of $H$ in $G$ such that $u_L$ can be identified with $i$ and $u_R$ with $j$.
    Equivalently,
    \[
        M_H(G)_{i,j} \coloneqq \sum_{\substack{\tau:V(H)\to[n] \\ \tau \text{ injective} \\ \tau(u_L) = i,~\tau(u_R) = j}} \prod_{\{a,b\}\in E(H)} G_{\tau(a),\tau(b)}
    \]
\end{definition}

\begin{remark}
    We only consider shapes $H$ whose size does not grow with $n$.
\end{remark}

Let us explain at this point how we explicitly construct the $r$-dimensional subspace. The idea is to lift the $r$ right eigenvectors $\psi_1, \ldots, \psi_r \in \R^{k}$ of $T$ above the KS threshold using the indicator vectors $\Ind_c$ for $c \in [k]$. This makes the computation tractable because edges in the SBM are conditionally independent given the true community assignment.  
Accordingly, we are interested in polynomials of the form:
\[
    \Ind_c^{\top} M_H(G) \Ind_{c'},
\]
where $c$ and $c'$ are colors in $[k]$ and $\Ind_c$ is the indicator vector of vertices with color $c$.
The above quadratic form counts the number of copies of $H$ inside $G$ where the vertex identified with $u_L$ has color $c$ and the vertex identified with $u_R$ has color $c'$.

We use the following special case of \cite[Theorem 6.6]{BMR21}.
\begin{lemma}
    Suppose $H$ is a connected graph, then for $\bG\sim\SBM_n(\Model,\pi, d)$, we have:
    \[
        \Ind_c^{\top} M_H(\bG) \Ind_{c'} = \sum_{\substack{\sigma:V(H)\to[k] \\ \sigma(u_L) = c, \sigma(u_R) = c'}} \prod_{v\in V(H)} \pi_{\sigma(v)} \prod_{ij\in H} M_{\sigma(i)\sigma(j)} \cdot d^{|E(H)|} \cdot n^{|E(H)|-|V(H)|} \pm o(n)
    \]
    with probability $1-o_n(1)$, where the $o\parens*{\cdot}$ hides factors depending on $|E(H)|$.
\end{lemma}

\begin{remark}  \label{rem:tree-v-nontree}
    Observe that in the above if $H$ is not a tree, then $|\Ind_c^{\top} M_H(\bG) \Ind_{c'}| = o(n)$.
    If $H$ is a tree, then it can be expressed as a path $P = p_0,\dots,p_{\ell}$ where $p_0 = u_L$ and $p_\ell = u_R$, along with a subtree per vertex on $P$.
    In this case, with high probability:
    \[
        \Ind_c^{\top} M_H(\bG) \Ind_{c'} = \angles*{e_c, T^{\ell} e_{c'}}_{\pi} \cdot d^{|V(H)|-1} \cdot n \pm o(n).
    \]
\end{remark}
More generally, we are interested in understanding the quadratic form of vectors that are constant on each color class with $M_G(H)$, i.e. lifts of vectors in $\R^k$.
\begin{corollary}   \label{cor:quad-form-tree-shapes}
    Let $H$ be a tree such that $u_L$ and $u_R$ are distance-$\ell$ apart.
    For any $f:[k]\to\R$, define $\boldf^{(n)}\coloneqq \sum_{c\in[k]} f(c)\cdot\Ind_c$. 
    Then with high probability, for all $f,g:[k]\to\R$ simultaneously:
    \[
      \angles{\boldf^{(n)}, M_H(\bG)\bg^{(n)}} = \angles*{f, T^{\ell} g}_{\pi} \cdot d^{|V(H)|-1}\cdot n \pm o(n).
    \]
\end{corollary}
The bulk of the contribution to the value of the quadratic form comes from tree shapes.
\begin{corollary}   \label{cor:only-trees-matter}
    With high probability over $\bG\sim\SBM_n(\Model,\pi, d)$, for any collection of shapes $\calH$ along with coefficients $(c_H)_{H\in\calH}$, and any functions $f,g:[k]\to\R$, we have:
    \[
        \angles*{\boldf^{(n)}, \sum_{H\in\calH} c_H M_H \bg^{(n)} } = \angles*{\boldf^{(n)}, \sum_{\substack{ H\in\calH \\ H~\text{ tree} }} c_H M_H \bg^{(n)} } \pm o(n).
    \]
\end{corollary}

To lighten the notation, for the rest of this section we make the dependence on the original graph $\bG$ implicit. So for our algorithm, we would like to illustrate negative eigenvalues in the spectrum of $M_{\ell}(t) = A^{(\ell)} H(t) A^{(\ell)}$. 

Let $\psi_1,\dots,\psi_r$ be the right eigenvectors of $T$ corresponding to the largest nontrivial eigenvalue $\lambda_2(T) = \sqrt{\frac{1+\eps}{d}}$ for some $\eps > 0$. 
We will show that with high probability for all $v$ in $\mathrm{span}\qty{\psi_1,\dots,\psi_r}$, the quadratic form $\angles*{\bv^{(n)}, M_{\ell}(t)\bv^{(n)}}$ is substantially negative, where $t = \frac{1+\delta}{\sqrt{(1+\eps)d}}$ for some small enough $\delta$ we choose in posterity.
\begin{lemma}\label{lemma:lifted-outlier}
    For all $\ell$ sufficiently large and all $v\in\mathrm{span}\qty{\psi_1,\dots,\psi_r}$, $\angles*{\bv^{(n)}, M_{\ell}(t) \bv^{(n)}} < -\alpha n$, where $\alpha > 0$ is an absolute constant.
    Consequently, $M_{\ell}(t)$ has at least $r$ eigenvalues less than $-\alpha$.
\end{lemma}
\begin{proof}
We can write $M_{\ell}(t)$ as a weighted combination of a collection of shape matrices $\calH$.
The plan is to extract out all the tree shapes that occur in $\calH$ along with their weights and obtain a formula for the quadratic form using \Cref{cor:quad-form-tree-shapes}.

We expand $M_{\ell}(t) = A^{(\ell)} A^{(\ell)} - t A^{(\ell)}AA^{(\ell)} + t^2 A^{(\ell)}(D-I)A^{(\ell)}$.
Each term counts walks following certain simple rules.
\begin{observation} \label{obs:walks-counted-per-term}
    We can characterize the walks counted by each term as follows.
    \begin{enumerate}
        \item $A^{(\ell)} A^{(\ell)}[i,j]$ counts walks starting at $i$ and ending at $j$ with the phases (i) walk $\ell$ nonbacktracking steps, (ii) walk $\ell$ nonbacktracking steps.
        \item $A^{(\ell)}AA^{(\ell)}[i,j]$ counts walks with phases (i) walk $\ell$ nonbacktracking steps, (ii) walk one step, (iii) walk $\ell$ nonbacktracking steps.
        \item $A^{(\ell)}(D-I)A^{(\ell)}[i,j]$ counts walks with phases (i) walk $\ell+1$ nonbacktracking steps and then backtrack one step, (ii) walk $\ell$ nonbacktracking steps.
    \end{enumerate}
\end{observation}

\begin{observation} \label{obs:tree-shape-structure}
    For a walk to give rise to a tree shape, each step must either (i) visit a vertex not visited in the walk so far, or (ii) use an edge that has already been stepped on.
\end{observation}
Using this rule, we can deduce all the tree shapes that come from each term.
The tree shapes arising in $\calH$ turn out to have quite simple structure.
We use $H_{s,i,r}$ to denote the shape obtained by taking a length-$s$ path where both endpoints are distinguished, and attaching a length-$r$ path at vertex $i$.

\medskip

\noindent {\bf Analysis of $A^{(\ell)}A^{(\ell)}$.}
By \Cref{obs:walks-counted-per-term}
and \Cref{obs:tree-shape-structure}, any tree shape for this term must arise as a result of first walking $\ell$ self-avoiding steps, then backtracking $s$ steps, and then walking $\ell-s$ self-avoiding steps for $0\le s\le \ell$.
The shape this gives rise to is $H_{2\ell - 2s, \ell-s, s}$.
Therefore, by \Cref{cor:only-trees-matter}, with high probability, for all $v:[k]\to\R$ we have:
\[
    \angles*{ \bv^{(n)}, A^{(\ell)}A^{(\ell)} \bv^{(n)} } = \sum_{s=0}^{\ell} \angles*{\bv^{(n)}, M_{H_{2\ell-2s, \ell-s, s}} \bv^{(n)}} \pm o(n) = \sum_{s=0}^{\ell} \angles*{v, T^{2\ell-2s}v}_{\pi} \cdot d^{2\ell-s} \cdot n \pm o(n).
\]

\medskip

\noindent {\bf Analysis of $A^{(\ell)} A A^{(\ell)}$.}
By \Cref{obs:walks-counted-per-term} and \Cref{obs:tree-shape-structure}, the tree shapes arising here must fall into one of the following cases.
\begin{enumerate}
    \item Length-$\ell$ self-avoiding walk, then $s$ backtracking edges, and then a length-$(\ell-s+1)$ self-avoiding walk for $0\le s\le \ell$.
    The shapes arising in this case are $H_{2\ell-2s+1, \ell-s, s}$.
    \item Length-$(\ell+1)$ self-avoiding walk, then $s$ backtracking edges, and then a length-$(\ell-s)$ self-avoiding walk for $1\le s\le \ell$.
    The shapes arising in this case are $H_{2\ell-2s+1, \ell+1-s, s}$.
    \item Length-$\ell$ self-avoiding walk using edge $e$ for its final step, then a backtracking step that uses $e$, and then a backtracking step that uses $e$ for a third time, followed by a self-avoiding walk of length-$(\ell-1)$.
    The shape arising in this case is simply a length-$(2\ell-1)$ path with both endpoints distinguished.
\end{enumerate}
By \Cref{cor:only-trees-matter}, with high probability, for all $v:[k]\to\R$:
\begin{align*}
    \angles*{\bv^{(n)}, A^{(\ell)} A A^{(\ell)} \bv^{(n)}} &= \parens*{\angles*{v, T^{2\ell-1}v} \cdot d^{2\ell-1} - \angles*{v, T^{2\ell+1}v} \cdot d^{2\ell+1}} \cdot n \, + \\
    &\sum_{s=0}^{\ell} 2\angles*{ v, T^{2\ell-2s+1} v } \cdot d^{2\ell-s+1} \cdot n \pm o(n).
\end{align*}

\medskip

\noindent {\bf Analysis of $A^{(\ell)}(D-I)A^{(\ell)}$.}
By \Cref{obs:walks-counted-per-term} and \Cref{obs:tree-shape-structure}, the tree shapes fall into one of the cases below.
\begin{enumerate}
    \item Length-$(\ell+1)$ self-avoiding walk, then $s$ nonbacktracking steps, followed by a length-$(\ell-s+1)$ self-avoiding walk for $1\le s \le \ell+1$.
    The shapes arising in this case are $H_{2\ell-2s+2, \ell-s+1, s}$.
    \item Length-$(\ell+1)$ self-avoiding walk where the last used edge is $e$, a backtracking step that uses $e$, and then a backtracking step that uses $e$ for a third time, followed by a length-$(\ell-1)$ self-avoiding walk.
    This shape is simply a length-$2\ell$ path with both endpoints distinguished.
\end{enumerate}
By \Cref{cor:only-trees-matter}, with high probability, for all $v:[k]\to\R$:
\begin{align*}
    \angles*{\bv^{(n)}, A^{(\ell)}(D-I)A^{(\ell)} \bv^{(n)}} = \angles*{v, T^{2\ell}v} \cdot d^{2\ell} \cdot n + \sum_{s=1}^{\ell+1} \angles*{v, T^{2\ell-2s+2} v} \cdot d^{2\ell-s+2} \cdot n \pm o(n).
\end{align*}

\medskip

\noindent {\bf Estimating $\angles*{\bv^{(n)}, M_{\ell}(t) \bv^{(n)} }$.}
By combining the above estimates, the value of $\angles*{\bv^{(n)}, M_{\ell}(t) \bv^{(n)} }$ is:
\(
    n\cdot\angles*{v, p(T) v} \pm o(n)
\)
where $p$ is the univariate polynomial
\[
    p(\lambda) \triangleq -t \cdot d^{2\ell-1} \lambda^{2\ell-1} \cdot (1 - \lambda^2 d^2) + t^2 d^{2\ell} \lambda^{2\ell} + \sum_{s=0}^{\ell} d^{2\ell-s} \lambda^{2\ell-2s} - 2t d^{2\ell-s+1} \lambda^{2\ell-2s+1} + t^2 d^{2\ell-s+1}\lambda^{2\ell-2s}.
\]
To prove that $\angles*{v,p(T)v}$ is negative for all $v$ in $\mathrm{span}\qty{\psi_1,\dots,\psi_r}$, it suffices to show it is negative for $v = \psi_1,\dots,\psi_r$ since those are also eigenvectors of $p(T)$.
In particular, it suffices to show that $p(\lambda_2)$ is negative.
Towards doing so, we first simplify the expression for $p(\lambda)$.
\begin{align*}
    p(\lambda) = (\lambda d)^{2\ell} \cdot \parens*{ t^2 - \frac{t}{\lambda d}\cdot (1-\lambda^2 d^2) + \parens*{ 1 - 2\lambda dt + dt^2 } \cdot \parens*{ \frac{1-(\lambda^2 d)^{-\ell-1}}{1-(\lambda^2 d)^{-1}} } }.
\end{align*}
We plug in $\lambda = \sqrt{\frac{1+\eps}{d}}$, $\ell = K \cdot \left( \frac{\log(1/\eps)}{\eps} \vee 1 \right)$ for some large $K > 0$, and $t = \frac{1+\delta}{\sqrt{(1+\eps)d}}$ for some small $\delta > 0$ we choose in posterity, which yields:
\begin{align*}
    &(\lambda d)^{2\ell} \cdot \parens*{ \frac{(1+\delta)^2}{(1+\eps)d} - \frac{1+\delta}{(1+\eps) d} +  (1+\delta) + \parens*{1 - 2(1+\delta) + \frac{(1+\delta)^2}{1+\eps}}\cdot\frac{1+\eps}{\eps} + o_K(1) } \\
    &= (\lambda d)^{2\ell} \cdot \parens*{ \frac{\delta+\delta^2}{(1+\eps)d} + (1+\delta) + \frac{ 1+\eps - 2(1+\delta)(1+\eps) + (1+\delta)^2 }{\eps} + o_K(1)} \\
    &= (\lambda d)^{2\ell} \cdot \parens*{ \frac{\delta+\delta^2}{(1+\eps)d} + (1+\delta) + \parens*{-1 - 2\delta + \frac{\delta^2}{\eps}} + o_K(1) } \\
    &= (\lambda d)^{2\ell} \cdot \parens*{-\delta + \frac{\delta}{(1+\eps)d} + \frac{\delta^2}{(1+\eps)d} + \frac{\delta^2}{\eps} + o_K(1)},
\end{align*}
where $o_K(1)$ denotes a quantity which vanishes as $K \to \infty$ regardless of the other parameters. The above expression can be made negative by choosing $\delta$ small enough, $K$ large enough, and using that $d\ge 1$.
Since $\lambda_2 = \sqrt{\frac{1+\eps}{d}}$ for some $\eps > 0$, for $\ell = O\left(\tfrac{\log(1/\eps)}{\eps}\right)$ there exists a choice of $\delta$ small enough and an absolute constant $\alpha > 0$ such that $p(\lambda_2) \le -\alpha$.
Consequently, for any $v$ in $\mathrm{span}\qty{\psi_1,\dots,\psi_r}$, $\angles*{v,p(T)v} \le -\alpha$.
\end{proof}

\input{outlier-after-truncation}

%% file: outlier-after-truncation.tex
\subsection{Outlier eigenspace after degree truncation}\label{sec:truncation}
In this section we show that by picking the truncation threshold $B$ to be large enough, the $r$-dimensional subspace which witnesses the negative eigenvalues for $M_{\bG, \ell}(t)$ also does the same for the truncated matrix $\ol{M}_{\bG_B, \ell}(t)$. In this section, we fix $t$ as specified by \cref{thm:main-outlier} and suppress dependence on $t$, as the same arguments apply regardless of whether we consider positive or negative $t$. 
\begin{lemma}[Truncation doesn't affect negative witnesses]\label{lemma:truncation}
    Let $\bG \sim \SBM_n(\Model, \pi, d)$. For $B > 0$ sufficiently large, the following holds true with probability $1-o_n(1)$. For all unit $x \in \mathrm{span}\qty{\psi_1, \ldots, \psi_r}$, we have $\ev{\bx^{(n)}, \ol{M}_{\bG_B, \ell} \bx^{(n)}} < -\alpha_{\scaleto{\ref{lemma:truncation}}{5pt}} n$, where $\alpha_{\scaleto{\ref{lemma:truncation}}{5pt}} > 0$ is an absolute constant. 
\end{lemma}

To facilitate the discussion, we introduce the following notion of an \emph{affected} vertex.
\begin{definition}[Affected vertices]
Let $G$ be a graph. We say that a vertex $v \in V(G)$ is \emph{affected} if its $2\ell+1$ neighborhood differs in $G$ and $G_B$, otherwise say that $v$ is \emph{unaffected}.  
\end{definition}

We also need the following concentration inequality, which follows from \cite[Lemma 29]{BLM15}. 
\begin{lemma}\label{lemma:sbm-neighborhoods}
    Let $\bG \sim \SBM_n(\Model, \pi, d)$. For any vertex $v$, let $\cN_{\ell}(v)$ denote the distance-$\ell$ neighborhood of $v$ in $\bG$. Then for any $s \ge 0$, there exists universal constants $C, c > 0$ such that 
    \begin{align*}
        \Pr[\exists \ell: \abs{\cN_{\ell}(v)} \ge s d^{\ell}]\le C\exp(-cs). 
    \end{align*}
\end{lemma}
In other words, the neighborhoods in SBMs look approximately regular. With these ingredients assembled, we can now prove \cref{lemma:truncation}. 

\begin{proof}[Proof of \cref{lemma:truncation}]
    The desired statement follows from \Cref{lemma:lifted-outlier} once we prove $\abs{\ev{uu^\top, M_{\bG, \ell} - \ol{M}_{\bG_B, \ell}}}\le o_B(1)$, and choose $B$ as a large enough constant.

    Throughout the proof, we use the following crude estimates: if $\abs{\cN_{2\ell+1}(v)} \le L$, then $\norm{M_{\bG, \ell}[v]}_1 \le 2L$, and $\norm{\ol{M}_{\bG_B, \ell}[v]}_1 \le B^{2\ell+1}$.

    By H\"{o}lder's inequality we have 
    \begin{align*}
        \abs{\ev{xx^\top, M_{\bG, \ell} - \ol{M}_{\bG_B, \ell}}} \le \norm{xx^\top}_{\max} \norm{M_{\bG, \ell} - \ol{M}_{\bG_B, \ell}}_1,
    \end{align*}
    where $\norm{M_{\bG, \ell} - \ol{M}_{\bG_B, \ell}}_1$ is the entrywise $\ell_1$ norm. Let $v$ be any vertex in $\bG$; we first upper bound the expected contribution of the $v$th row of $M_{\bG, \ell} - \ol{M}_{\bG_B, \ell}$ to the norm. 

    By our earlier observation and triangle inequality, we deduce 
    \begin{align*}
        \norm{M_{\bG, \ell} - \ol{M}_{\bG_B, \ell}}_1 &\le 2\sum_{v} \bone[v \text{ is affected}] (\abs{\cN_{2\ell+1}(v)} + B^{2\ell+1}) \\
        &= 2\sum_{v} \bone[\deg(v) > B] \sum_{w \in \cN_{2\ell+1}(v)} (\abs{\cN_{2\ell+1}(w)} + B^{2\ell+1}) \\ 
        &\le 2\sum_{v} \bone[\deg(v) > B] \abs{\cN_{4\ell+2}(v)}(\abs{\cN_{4\ell+2}(v)} + B^{2\ell+1}),
    \end{align*}
    since $\cN_{4\ell+2}(v) \supseteq \cN_{2\ell+1}(w)$. For ease of notation, introduce the quantity $f(v) \triangleq \abs{\cN_{4\ell+2}(v)}(\abs{\cN_{4\ell+2}(v)} + B^{2\ell+1})$. Taking expectations, we see that 
    \begin{align*}
        \expt{\norm{M_{\bG, \ell} - \ol{M}_{\bG_B, \ell}}_1} &\le 2 \sum_{v} \Pr[\deg(v) > B] \E[f(v) | \deg(v) > B] 
    \end{align*}
    To proceed, we use the general property that if $\mathcal{E}$ is an event, $X$ is a random variable with CDF $\Phi$, and $f: \R \to \R$ is increasing, then $\E[f(X) | \mathcal{E}] \le \E[f(X) | X \ge \Phi^{-1}(\Pr[\mathcal{E}])]$. Take $\mathcal{E}$ to be the event $\qty{\deg(v) > B}$, and set $p \triangleq \Pr[\deg(v) > B]$. Then as $\E[X | X \ge \Phi^{-1}(p)] = \frac{\E[X \bone[X \ge \Phi^{-1}(p)]]}{p}$, we can upper bound the conditional expectation by inverting \cref{lemma:sbm-neighborhoods}. We see that 
    \begin{align*}
        \Pr[\deg(v) > B] \E[f(v) | \deg(v) > B] &\le \int_{c\log(C/p)}^{\infty} yd^{4\ell+2}(yd^{4\ell+2} + B^{2\ell+1})\Pr[\abs{\cN_{4\ell+2}(v)} \ge yd^{4\ell+2}]dy \\
        &\le C\int_{c\log(C/p)}^{\infty} yd^{4\ell+2}(yd^{4\ell+2} + B^{2\ell+1}) \exp(-cy)dy \\
        &\le O(B^{8\ell+4}) \cdot p(\log(1/p)^2 + \log(1/p)), 
    \end{align*}
    where the last line follows from from an explicit integral computation. Note that $p(\log(1/p)^2 + \log(1/p))$ is monotonically increasing for $p < 0.1$, so for sufficiently large $B$ we can simply set $p$ to be the upper bound furnished by \cref{lemma:sbm-neighborhoods}. Indeed, since $\ell = O(1)$, we conclude that 
    \begin{align*}
        \Pr[\deg(v) > B] \E[f(v) | \deg(v) > B] &\le O(B^{8\ell+6}) \exp(-cB/d) \\
        &= o_B(1),
    \end{align*}
    for $B = \Omega(\ell \log \ell)$. To summarize, we have just shown that
    \begin{align*}
        \expt{\norm{M_{\bG, \ell} - \ol{M}_{\bG_B, \ell}}_1} = o_B(1)n
    \end{align*}

    Next, we bound the variance. Since $\ell = O(1)$, an Efron-Stein argument as before shows that the variance of $\norm{M_{\bG, \ell} - \ol{M}_{\bG_B, \ell}}_1$ is $\widetilde{O}(n)$. More precisely, consider rerandomizing edge $(i,j) \in \bG$. With probability $1 - O(1/n)$, the rerandomized edge agrees with the original edge, in which case the $\ell_1$ norm doesn't change. Otherwise, with probability $O(1/n)$, we can crudely bound the effect on the $\ell_1$ norm by $(\log n)^{\ell}$. Hence $\Var\qty(\norm{M_{\bG, \ell} - \ol{M}_{\bG_B, \ell}}_1) \le \widetilde{O}(n)$, and the desired concentration follows.
\end{proof}

As a corollary, we obtain the analogous version of \cref{thm:main-outlier} for the truncated matrix $\ol{M}_{\bG_B, \ell}$.
\begin{lemma}[Spectral gap for truncated matrix]\label{lemma:spectral-gap}
For sufficiently large $B > 0$, with high probability over $\bG \sim \SBM_n(\Model, \pi, d)$, the matrix $\ol{M}_{\bG_B, \ell}$ has at most $r+1$ negative eigenvalues, at least $r$ of which are upper bounded by $-\alpha_{\scaleto{\ref{lemma:truncation}}{5pt}}$.
\end{lemma}
\begin{proof}
    By construction (see \cref{sec:recovery-algo}), $\ol{H}_{\ell}$ is a principal submatrix of $H_{\ell}$. This is where we use the non-standard way in which we truncated the degree matrix to get $\ol{H}_{\ell}$. It would not have been true if we had naively truncated $D_{\bG_B}$. 

    Since $H_{\ell}$ has at most $r+1$ negative eigenvalues, by the Cauchy interlacing theorem, its principal submatrix $\ol{H}_{\ell}$ has at most $r+1$ negative eigenvalues. 
    The number of negative eigenvalues of $ \ol{M}_{\bG_B,\ell} = A_{\bG_B}^{(\ell)} \ol{H}_{\ell} A_{\bG_B}^{(\ell)}$ is the same as that of $\ol{H}_{\ell}$.  
    Hence $\ol{M}_{\bG_B, \ell}$ has at most $r+1$ negative eigenvalues, and \cref{lemma:truncation} implies that $\ol{M}_{\bG_B, \ell}$ has at least $r$ negative eigenvalues.
    Now using $\alpha_{\scaleto{\ref{lemma:truncation}}{5pt}}$ guaranteed by the second part of \cref{lemma:truncation}, we conclude.
\end{proof}

%% file: robust-eigenspace-recovery.tex
\section{Robust recovery of a subspace} \label{sec:robust-eigenspace}

In this section, we consider the following algorithmic problem, which we recall from \cref{sec:robust-eigenspace-overview}. 
\begin{assumption}  \label{assumption:matrix-perturb}
    Let $M$ be a $n \times n$ matrix with at most $r$ negative eigenvalues.
    Let $\wt{M}$ be a corrupted version of $M$.
    Formally, let $\wt{M} = M + \Delta$ where $\Delta$ is supported on a $\gamma n \times \gamma n$ submatrix, indexed by $Q\subseteq[n]$.
    Further, assume that the $\ell_1$-norm of every row and column of $M$ and $\Delta$ are bounded by $K$.
    Observe that this implies that $\norm*{M} \leq K$ and $\norm*{\Delta} \leq K$.
\end{assumption}

\problemsubspace*

We show that Phase 2 of the algorithm described in \Cref{sec:recovery-algo} with $\wt{M}$ as input indeed solves \Cref{prob:subspace-recovery}, where a more precise quantitative dependence is articulated below.
\begin{theorem} \label{thm:eigenspace-recovery}
    For every $C, v, K > 0$, the following holds for all sufficiently small $\gamma > 0$:
    There is an efficient algorithm that takes in $\wt{M}$ from \Cref{assumption:matrix-perturb} along with $C$ and $\upsilon$ as input, and with high probability over the randomness of the algorithm, outputs a subspace $U$ with the following properties.
    \begin{enumerate}
        \item \label{item:low-dim} The dimension of $U$ is at most $O\parens*{\frac{K}{\upsilon}\cdot r}$.
        \item \label{item:deloc} The diagonal entries of the projection matrix $\Pi_U$ are at most $O\parens*{\frac{C^2K^3}{\upsilon^3}\cdot\frac{r}{n}}$.
        \item \label{item:deloc-quadform} For any $C$-delocalized unit vector $y$ such that $\angles*{y,My} < -\upsilon$, we have $\angles*{y, \Pi_U y} \ge \Omega\parens*{\frac{\upsilon}{K}}$.
    \end{enumerate}
\end{theorem}

To prove \Cref{thm:eigenspace-recovery}, we analyze Phase 2 of the algorithm in \Cref{sec:recovery-algo}.
The precise version of the algorithm we analyze is stated below.

\noindent\rule{16cm}{0.4pt}

\begin{algorithm}   \label{algo:subspace-recovery}
The provided input is $\wt{M}$ according to \Cref{assumption:matrix-perturb}, and real numbers $\upsilon$ and $C$.

Define $\eta \triangleq \frac{\upsilon}{48}$.
\begin{enumerate}
    \item Define $\wt{M}^{(0)}$ as $\wt{M}$.  Let $t$ as a counter initialized at $0$, and let $\Phi(X)$ be the number of eigenvalues of $X$ smaller than $-\eta$.
    \item While $\Phi(\wt{M}^{(t)}) > \frac{2K}{\eta}r$: compute the projection matrix $\Pi^{(t)}\triangleq \Pi_{\le -\eta}(\wt{M}^{(t)})$, choose a random $i\in[n]$ with probability $\frac{\Pi^{(t)}_{i,i}}{\Tr\parens*{\Pi^{(t)}}}$, and define $\wt{M}^{(t+1)}$ as the matrix obtained by zeroing out the $i$-th row and column of $\wt{M}^{(t)}$.
    Then increment $t$.
    \item Let $T$ be the time of termination and let $\tau = \frac{2C^2K^2r}{\eta^2}$.
    Compute $\Pi^{(T)}$, and define the set of indices $S$ as,
    \[ S = \left\{ i \Big |  \Pi^{(T)}_{i,i} \le \frac{\tau}{n} \right\} \ .  \]
    Define $\wt{\Pi}$ as $\parens*{\Pi^{(T)}_{S,S}}_{\ge \eta/K}$, and compute its span $U$.\footnote{Recall that $\parens*{X}_{\ge a}$ denotes the truncation of the eigendecomposition of $X$ for eigenvalues at least $a$.}
    \item Output $U$.
\end{enumerate}
\noindent\rule{16cm}{0.4pt}
\end{algorithm}

We first bound on the number of eigenvalues of $\wt{M}^{(T)}$ smaller than $-\eta$, thus establishing \Cref{item:low-dim} of \Cref{thm:eigenspace-recovery}.

\subsection{Proof of \Cref{item:low-dim}: spectrum of cleaned-up matrix}\label{subsec:cleaned-up-spectrum}

\begin{lemma}   \label{lem:cleaned-up-matrix}
    The matrix $\wt{M}^{(T)}$ has the following properties:
    \begin{enumerate}
        \item \label{item:few-outliers} $\wt{M}^{(T)}$ has at most $\frac{2K}{\eta}r$ eigenvalues smaller than $-\eta$.
        \item \label{item:small-delete} $\wt{M}^{(T)}$ is a submatrix $\wt{M}_{R,R}$ where $|R|\ge \parens*{1-\frac{4K}{\eta}\gamma}n$ with high probability over the randomness of the algorithm.
    \end{enumerate}
\end{lemma}

In service of proving \Cref{lem:cleaned-up-matrix}, we prove the following statement about the localization of outlier eigenvectors on the indices of the corruptions $Q$.
\begin{lemma}   \label{lem:outliers-localize}
    Let $R^{(t)} \subseteq [n]$ denote the non-zero rows and columns of $\wt{M}^{(t)}$ and let $Q^{(t)} = Q \cap R^{(t)}$. Then the following holds:
    \[
        \frac{\Tr\parens*{\Pi_{\le -\eta}\parens*{\wt{M}^{(t)}}_{Q^{(t)}, Q^{(t)}}}}{\Tr\parens*{\Pi_{\le -\eta}\parens*{\wt{M}^{(t)}}}} \ge \frac{\eta}{K} - \frac{r}{\Tr\parens*{\Pi_{\le -\eta}\parens*{\wt{M}^{(t)}}}}
    \]
\end{lemma}
\begin{proof}
    The desired statement follows from the chain of inequalities below.
    \begin{align*}
        -\eta \cdot \Tr\parens*{\Pi_{\le-\eta}\parens*{\wt{M}^{(t)}}} &\ge \Tr\parens*{\wt{M}^{(t)}_{\le-\eta}} \\
        &= \angles*{\wt{M}^{(t)}, \Pi_{\le-\eta}\parens*{\wt{M}^{(t)}}} \\
        &= \angles*{M_{R^{(t)},R^{(t)}}+ \Delta_{R^{(t)},R^{(t)}}, \Pi_{\le-\eta}\parens*{\wt{M}^{(t)}}} \\
        &\ge \Tr(\parens*{M_{R^{(t)},R^{(t)}}}_{\le 0}) - \norm*{\Delta} \cdot \Tr\parens*{ \Pi_{\le-\eta}\parens*{\wt{M}^{(t)}}_{Q^{(t)}, Q^{(t)}} } \\
        &\ge -Kr - K\cdot \Tr\parens*{ \Pi_{\le-\eta}\parens*{\wt{M}^{(t)}}_{Q^{(t)}, Q^{(t)}} }.
    \end{align*}
    In the last step, we used the fact that $\norm{M_{R^{(t)},R^{(t)}}} \leq \norm{M} \leq K$ and the number of negative eigenvalues of $\parens*{M}_{R^{(t)},R^{(t)}}$ is at most the number of negative eigenvalues of $M$, namely $r$.
\end{proof}

We now prove \Cref{lem:cleaned-up-matrix}.
\begin{proof}[Proof of \Cref{lem:cleaned-up-matrix}] 
    By definition of the iterative process, \Cref{item:few-outliers} is satisfied on termination, and the process has to terminate since it cannot run for more than $n$ steps.

    It remains to prove \Cref{item:small-delete}.
    Towards doing so, first observe that the process terminates if all rows and columns in $Q$ are zeroed out.
    Indeed, $M$ itself satisfies the termination condition, and if all the rows and columns in $Q$ are zeroed out, then the support is a principal submatrix of $M$, which satisfies the termination condition by Cauchy's interlacing theorem.

    Suppose the termination condition is not satisfied at iteration $t$, then the $i$-th row/column of the matrix is zeroed out where $i$ is chosen with probability $\frac{\Pi_{\le-\eta}\parens{\wt{M}^{(t)}}_{i,i}}{\Tr\parens*{\Pi_{\le-\eta}\parens{\wt{M}^{(t)}}}}$.
    Consequently, the probability that $i$ is in $Q$ is equal to:
    \[
        \frac{\Tr\parens*{\Pi_{\le -\eta}\parens*{\wt{M}^{(t)}}_{Q^{(t)}, Q^{(t)}}}}{\Tr\parens*{\Pi_{\le -\eta}\parens*{\wt{M}^{(t)}}}},
    \]
    which by \Cref{lem:outliers-localize} is at least
    \[
        \frac{\eta}{K} - \frac{r}{\Tr\parens*{\Pi_{\le -\eta}\parens*{\wt{M}^{(t)}}}}.
    \]
    Since $\Tr\parens*{\Pi_{\le -\eta}\parens*{\wt{M}^{(t)}}}$ counts the number of eigenvalues smaller than $-\eta$ and the termination condition is not satisfied, the above probability is at least $\frac{\eta}{2K}$.

    Since the process terminates if all of $Q$ is zeroed out (in which case $R = \ol{Q}$), a standard martingale argument shows that the process can last for at most $\frac{4K}{\eta}\gamma n$ steps with high probability.
\end{proof}

We now proceed to proving \Cref{item:deloc,item:deloc-quadform} of \Cref{thm:eigenspace-recovery}.

\subsection{Proofs of \Cref{item:deloc,item:deloc-quadform}: delocalization and correlation with recovered subspace}\label{subsec:postproc-props}
The proof of \Cref{item:deloc} is fairly straightforward.
\begin{lemma}
    All the diagonal entries of $\Pi_U$, the projection matrix onto $U$, are bounded by $\frac{2C^2K^3}{\eta^3}\cdot\frac{r}{n}$.
\end{lemma}
\begin{proof}
    Let $S$ be the set of indices whose diagonal entries are at most $\frac{\tau}{n}$, where $\tau = \frac{2C^2K^2}{\eta^2}\cdot r$.
    By definition, $\Pi_{S,S}$ has entries bounded by $\frac{\tau}{n}$, and since $\Pi_{S,S}\psdge\parens*{\Pi_{S,S}}_{\ge\frac{\eta}{K}} = \wt{\Pi}$, the diagonal entries of $\wt{\Pi}$ are at most $\frac{\tau}{n}$.
    All eigenvalues of $\wt{\Pi}$ are at least $\frac{\eta}{K}$, and hence $\frac{K}{\eta}\cdot\wt{\Pi}\psdge \Pi_U$, which gives us a bound of $\frac{K\tau}{\eta n}$ on the diagonal entries of $\Pi_U$, from which the statement follows.
\end{proof}

Now, we turn our attention to proving \Cref{item:deloc-quadform}.
We first show that $\wt{M}_{R,R} = \wt{M}^{(T)}$ preserves negative correlations with delocalized vectors.
Henceforth, we assume that the high probability guarantee from \Cref{lem:cleaned-up-matrix} that $|R|\ge \parens*{1-\frac{4K}{\eta}\gamma}n$ holds.
\begin{lemma}   \label{lem:neg-cleanup}
    For unit $C$-delocalized vector $y$ such that $\angles*{y, My} < -\upsilon$, we have:
    \(
        \angles*{y, \wt{M}_{R,R}y} < -\frac{\upsilon}{2}.
    \)
\end{lemma}
\begin{proof}
    By \Cref{assumption:matrix-perturb,lem:cleaned-up-matrix}, we can write $\wt{M}_{R,R}$ as $M+E_1+E_2$ where $E_1$ is supported on at most $\parens*{\frac{4K}{\eta}+1}\gamma n$ rows and $E_2$ is supported on at most $\parens*{\frac{4K}{\eta}+1}\gamma n$ columns.
    Consequently, the entrywise $\ell_1$ norm of $E \triangleq E_1+E_2$ is at most $\frac{10K^2}{\eta}\gamma n$.
    The statement follows from the inequality below for small enough $\gamma$:
    $$
     \angles*{y, Ey} \le \norm*{yy^{\top}}_{\max}\cdot\norm*{E}_1 \le O\parens*{ \frac{10K^2C^2}{\eta} \gamma } < \frac{\upsilon}{2}. \qedhere
    $$
\end{proof}

We have the following immediate corollary of \Cref{lem:neg-cleanup}.
\begin{corollary}   \label{cor:cleanup-proj-corr}
    For any $C$-delocalized unit vector $y$ such that $\angles*{y, My} < -\upsilon$, we have:
    \[
        \angles*{y, \Pi_{\le -\eta}\parens*{\wt{M}_{R,R}} y } \ge \frac{\upsilon/2 - \eta}{2K-\eta}.
    \]
\end{corollary}
\begin{proof}
Starting with \Cref{lem:neg-cleanup}, we get
\begin{align*}
-\frac{v}{2} & > \angles*{y, \wt{M}_{R,R} y}\\
             &\geq -\norm{\wt{M}_{R,R}} \cdot \angles*{y, \Pi_{\le -\eta}\parens*{\wt{M}_{R,R}} y } + (-\eta) \cdot \parens*{1- \angles*{y, \Pi_{\le -\eta}\parens*{\wt{M}_{R,R}} y }} \ .
\end{align*}
Since the spectral norm of a matrix is at most its maximum $\ell_1$ norm of a row, $\norm*{\wt{M}_{R,R}} \le 2K$.  Substituting in the above inequality and simplifying yields the desired result.
\end{proof}

We are finally ready to establish \Cref{item:deloc-quadform}.
\begin{lemma}
    For any unit $C$-delocalized vector $y$ such that $\angles*{y,My} < -\upsilon$,
    \[
         \angles*{y, \Pi_U y} \ge \Omega\parens*{\frac{\upsilon}{K}}.
    \]
\end{lemma}
\begin{proof}
    We abbreviate $\Pi_{\le-\eta}\parens*{\wt{M}_{R,R}}$ to $\Pi$.
    Recall $\tau = \frac{2C^2K^2r}{\eta^2}$ from \Cref{algo:subspace-recovery}, and let $S$ denote the set of indices where the diagonal entries of $\Pi$ are at most $\frac{\tau}{n}$.
    By \Cref{lem:cleaned-up-matrix}, $\Tr(\Pi)\le \frac{2K}{\eta}r$, and hence
    \(
        |\ol{S}| \le \frac{2Kr}{\eta\tau}n
    \).
    By positivity of $\Pi$ and \Cref{cor:cleanup-proj-corr},
    \[
        \angles*{ y, \Pi_{S,S} y} + \angles*{ y, \Pi_{\ol{S},\ol{S}} y } \ge \frac{1}{2} \angles*{ y, \Pi y} \ge \frac{1}{2}\cdot \frac{\upsilon/2-\eta}{2K-\eta}.
    \]
    The second term of the LHS is equal to
    $$
        \angles*{yy^{\top}_{\ol{S},\ol{S}}, \Pi} \le \Tr\parens*{yy^{\top}_{\ol{S},\ol{S}}} \cdot \norm*{\Pi} \le \frac{C^2 |\ol{S}|}{n} \le \frac{\eta}{K},
    $$
    and consequently,
    \[
        \angles*{ y, \Pi_{S,S} y } \ge \frac{1}{2} \cdot \frac{\upsilon/2 - \eta}{2K-\eta} - \frac{\eta}{K}.
    \]
    Finally, since $\wt{\Pi}$ is obtained by discarding eigenvalues of $\Pi_{S,S}$ of magnitude at most $\frac{\eta}{K}$:
    \[
        \angles*{y, \wt{\Pi} y} \ge \frac{1}{2} \cdot \frac{\upsilon/2 - \eta}{2K-\eta} - \frac{2\eta}{K} \ge \frac{\upsilon/8 - 3\eta}{K} \ge \frac{\upsilon}{16K}.
    \]
    The statement follows since $\wt{\Pi}\psdle \Pi_U$.
\end{proof}

\subsection{Correlation of subspace with communities}
In the model of adversarial corruptions we are considering, a graph $\bG$ is sampled from $\SBM_n(\Model, \pi, d)$.
The graph then undergoes $\delta n$ adversarial edge insertions and deletions, and the resulting graph $\wt{\bG}$ is handed to us.

In this section, we prove a guarantee on the subspace $U$ produced by Phase 2 of the algorithm in \Cref{sec:recovery-algo} as a special case of \Cref{thm:eigenspace-recovery}.
\begin{lemma} \label{lem:subspace-output}
    The subspace $U$ produced as output at the end of phase 2 of the algorithm described in \Cref{sec:recovery-algo} has the following guarantees with high probability over the randomness of the stochastic block model.
    \begin{enumerate}
        \item The dimension of $U$ is at most $O\parens*{\frac{B^{2\ell+3}}{\alpha_{\scaleto{\ref{lemma:truncation}}{5pt}}}r}$.
        \item All the diagonal entries of $\Pi_U$ are bounded by $O\parens*{\frac{B^{6\ell+9}}{\alpha_{\scaleto{\ref{lemma:truncation}}{5pt}}^3\pi_{\min}} \cdot \frac{r}{n}}$.
        \item For any unit $x$ in $\mathrm{span}\{\psi_1,\dots,\psi_r\}$,
        \[
            \angles*{\bx^{(n)}, \Pi_U \bx^{(n)}} \ge \Omega\parens*{\frac{\alpha_{\scaleto{\ref{lemma:truncation}}{5pt}}}{B^{2\ell+3}}\cdot n}.
        \]
    \end{enumerate}
\end{lemma}

To prove \Cref{lem:subspace-output}, we plug in the following into the setup of \Cref{thm:eigenspace-recovery}.
\begin{displayquote}
    $M = \ol{M}_{\bG_B,\ell}$, $\Delta \triangleq  \ol{M}_{\wt{\bG}_B,\ell} - \ol{M}_{\bG_B,\ell}$,
    $K = 2B^{2\ell+3}$, $C = \frac{1}{\sqrt{\pi_{\min}}}$, $\gamma = 8B^{2\ell+2}\delta$, and $\upsilon = \alpha_{\scaleto{\ref{lemma:truncation}}{5pt}}$.
\end{displayquote}
Our proof of \Cref{lem:subspace-output} is complete once we verify that:
\begin{enumerate}
    \item $\frac{1}{\sqrt{n}}\bx^{(n)}$ is $C$-delocalized and achieves quadratic form less than $-\upsilon$ with our choice of $M$.
    \item The choice of $\gamma$ is indeed a sufficiently small constant depending on $\nu,C$, and $K$.
    \item $M$, $\Delta$, and $K$ satisfy \Cref{assumption:matrix-perturb}.
\end{enumerate}

\noindent {\bf Verifying the assumption on $\bx^{(n)}$.}
Note that \Cref{lemma:truncation} establishes that the quadratic form of $\bx^{(n)}$ is sufficiently negative, and the delocalization follows from the fact that $\norm*{\bx^{(n)}}_{\infty} \le \norm*{x}_{\infty} \le \frac{1}{\sqrt{\pi}_{\min}}$ since $x$ is a unit vector, along with the fact that $\norm*{\bx^{(n)}}$ concentrates around $\sqrt{n}$.

\medskip

\noindent {\bf Verifying the assumption on $\gamma$.}
The magnitude of $\gamma$ can be made sufficiently small by choosing $\delta$ as a sufficient small constant.

\medskip

\noindent {\bf Verifying the assumption on $M$, $\Delta$, and $K$.}
The below lemma proves why our choice of $M$, $\Delta$ and $K$ satisfy \Cref{assumption:matrix-perturb}.
\begin{lemma}   \label{lem:corruption-M}
    The matrix $\Delta \triangleq \ol{M}_{\bG_B,\ell} - \ol{M}_{\wt{\bG}_B,\ell}$ is supported on a $8 B^{2\ell + 2} \delta  n \times 8 B^{2\ell + 2} \delta n$ principal submatrix.
    Further, each row and column of $\ol{M}_{\bG_B,\ell}$ and $\ol{M}_{\wt{\bG}_B,\ell}$ have $\ell_1$ norm bounded by $B^{2\ell+3}$, and consequently each row and column of $\Delta$ has $\ell_1$ norm bounded by $2B^{2\ell+3}$.
\end{lemma}

\begin{proof}[Proof of \Cref{lem:corruption-M}]
    First, observe that $\ol{D}_{\bG_B}$ differs from $\ol{D}_{\wt{\bG}_B}$ only on rows and columns corresponding to vertices incident to a corrupted edge, of which there are at most $2\delta n$.

    Next, we bound the number of vertices whose neighborhoods differ between $\bG_B$ and $\wt{\bG}_B$.
    Observe that if a vertex $v$ has differing neighborhoods between $\bG_B$ and $\wt{\bG}_B$, then either $v$ is incident to an edge corruption, or $v$ has a neighbor in $\bG_B \cup \wt{\bG}_B$ that is truncated in one of $\wt{\bG}_B$ or $\bG_{B}$ but not the other.

    The number of vertices incident to an edge corruption is at most $2\delta n$.
    We can bound the number of vertices of the second kind by:
    \begin{align*}
        &\, \sum_{v\in V(\bG)} \sum_{u\in N_{\bG_B \cup \wt{\bG}_B }(v)} \abs*{ \Ind[\deg_{\bG}(u) > B] - \Ind[\deg_{\wt{\bG}}(u) > B] } \\
        &= \sum_{u\in V(\bG)} \sum_{v\in N_{\bG_B \cup \wt{\bG}_B }(u)} \abs*{ \Ind[\deg_{\bG}(u) > B] - \Ind[\deg_{\wt{\bG}}(u) > B] } \\
        &\le \sum_{u\in V(\bG)} B\cdot \abs*{ \Ind[\deg_{\bG}(u) > B] - \Ind[\deg_{\wt{\bG}}(u) > B] } \\
        &\le 2\delta B n
    \end{align*}
    where the last inequality is because $\deg_{\bG}(u)$ can differ from $\deg_{\wt{\bG}}(u)$ only when $u$ is incident to an edge corruption.
    Thus, the number of vertices with differing neighborhoods is at most $2\delta(B+1)n \le 4\delta B n$.

    Finally, $\ol{M}_{\wt{\bG}_B, \ell}$ and $\ol{M}_{\bG_B, \ell}$ can differ only on rows and columns corresponding to vertices that are at most $2\ell+1$ away in either $\wt{\bG}_B$ or $\bG_B$ from a vertex with differing neighborhoods.
    Since the degrees in both graphs are bounded by $B$, there are at most $8\delta B^{2\ell+3} n$ such vertices.

    The $\ell_1$ norm of any row $v$ of either $\ol{M}_{\wt{\bG}_B, \ell}$ or $\ol{M}_{\bG_B, \ell}$ is bounded by the total number of walks leaving $v$ of length at most $2\ell+2$, which is at most $B^{2\ell+3}$.
    Thus, the $\ell_1$ norm of any row of $\Delta$ is at most $2B^{2\ell+3}$, which completes the proof.
\end{proof}

%% file: rounding-spectral.tex
\section{Rounding algorithm}    \label{sec:rounding}
In this section we fill in the details for \cref{sec:algo-analysis}, which analyzed the rounding guarantees. First, we review the rounding procedure described in \cref{sec:recovery-algo}. From phase 2 of the algorithm, we obtain a constant-dimensional delocalized subspace $U$ which has constant correlation with the true communities. We process $U$ into a new matrix $M' \in \R^{n \times r'}$ which is directly used in the rounding scheme. Consider the matrix $\Psi_{r'}$ of the $r'$ right eigenvectors of $T$ corresponding to $\lambda_2$, assumed to be above the KS threshold:
\begin{align*}
    \Psi_{r'} \triangleq \mqty[\psi_1 & \cdots & \psi_{r'}].
\end{align*}
The rows of $\Psi_{r'}$ induce row embeddings $\{\phi_i\}_{i \in [k]}$ into $\R^{r'}$. We then rescale the rows of $M'$ by a uniform factor $c$ such that we can express each row of $c \cdot M'$ as a convex combination of the row embeddings $\{\phi_i\}_{i \in [k]}$. These weights are then used to sample a valid community assignment.

In \cref{lemma:convex-hull} we prove that the scaling $c$ used to obtain the rounding weights can be taken to be $\Omega(\sqrt{n})$ so long as the subspace $U$ from phase 2 of the algorithm is delocalized.
\begin{lemma}\label{lemma:convex-hull}
    Suppose that $\diag(\Pi_U) \le O(\frac{1}{n})$. 
    There exists $c = \Theta(\sqrt{n})$ such that the scaled rows of $c \cdot M'$ lie within the convex hull of the $\phi_i$'s. 
\end{lemma}
\begin{proof}
We first show that we can pick some $c = \Omega(1)$. Since $\psi_1 = \bone$ is the trivial eigenvector of $T$, by spectral theory we have that $\ev{\bone, \psi_j}_{\pi} = 0$ for $j \neq 1$. This implies that the convex combination $\sum_{i=1}^k \pi(i) \phi_i = 0 \in \R^{r-1}$, so the origin is in the convex hull of the row embeddings $\phi_i$. 
Note that this alone does not suffice to show that we can take $c > 0$ because the origin might be on the boundary of the hull. 

However, as we'll show in \cref{lemma:origin-interior} below, the origin in fact lies in the interior of the convex hull of the $\phi_i$'s, which permits us to take $c > 0$.  Furthermore, the minimum distance from the origin to the boundary is bounded below by a positive constant (possibly depending on $k$), because $0$ lies in the interior and the $\phi_i$'s are constant-sized objects. So we can take $c = \Omega(1)$.
To attain $c = \Theta(\sqrt{n})$, first recall that $\diag(\Pi_U) \le O(\frac{1}{n})$ by assumption, which implies that $\diag(\Pi_U') \le O(\frac{1}{n})$ since $U' \subseteq U$. By definition, $\Pi_{U'} = M' M'^\top$, the $\ell_2$ norm of each row of $M'$ is $O(\frac{1}{\sqrt{n}})$, from which the desired claim holds.
\end{proof}

\begin{lemma}\label{lemma:origin-interior}
    Assuming that $\pi > 0$, we have $0 \in \mathrm{int}(\mathrm{conv}(\phi_1, \ldots, \phi_k))$.
\end{lemma}
\begin{proof}
    First, note that since the column eigenvectors $\psi_j$ are orthogonal with respect to $\ev{\cdot, \cdot}_{\pi}$, and $\pi > 0$, the matrix of column eigenvectors has full row rank. On the other hand, we know that $0 \in \mathrm{conv}(\phi_1, \ldots, \phi_k)$ because $\sum_{i} \pi(i) \phi_i = 0$. Since the $\phi_i$ are linearly independent, they in fact lie on the boundary of the convex hull and the interior of the convex hull is a nonempty open set in $\R^r$. As $\pi > 0$, we have just shown that 0 lies in the interior of the convex hull.
\end{proof}


%% file: weak-correlation-definition.tex
\section{On defining weak-recovery }\label{sec:justify-metric} 

One might wonder whether our definition of weak recovery (\cref{def:weak-recovery}) is meaningful. 
In \cref{sec:weak-recovery-intuition}, we argue that our definition is natural and interpret it in various natural settings.
In \cref{sec:weak-recovery-comparison}, we show that our definition implies the other definitions of weak recovery found in the literature.

For self-containedness, we reproduce our and previous definitions of weak recovery below. 
\weakrecovery*

Next, we recall the notion of weak recovery used in \cite{DdNS22}. Note that their metric is specialized to the symmetric 2-community block model.
\begin{definition}[Weak recovery --- symmetric 2-community SBM]\label{def:simplified-overlap}
    We say that an algorithm achieves weak recovery for the symmetric 2-community SBM if it outputs a labeling $\wh{\bx} \in \qty{\pm 1}^n$ such that 
    \begin{align*}
        \E_{\bx, \bG} [|\langle\bx, \wh{\bx}\rangle|] \ge \Omega(n).
    \end{align*}
\end{definition}

Moving beyond the symmetric 2-community setting, we recall a more general notion of weak recovery present in the literature, which applies to arbitrary $k$-community SBMs \cite{Abb17}.
To define it, it is helpful to introduce the shorthand $\Omega_i = \qty{v \in [n]: \bX_v = i}$ for $i \in [k]$.
\begin{definition}[Weak recovery --- nontrivial partition]\label{def:partition}
    We say that an algorithm achieves weak recovery in the sense of nontrivial partition if it outputs with high probability a subset of vertices $S \subseteq V$ such that there exists communities $i, j \in [k]$ with 
    \begin{align*}
        \frac{|\Omega_i \cap S|}{|\Omega_i|} - \frac{|\Omega_j \cap S|}{|\Omega_j|} \ge \Omega(1).
    \end{align*}
\end{definition}

Note that \cref{def:partition} implies \cref{def:simplified-overlap}, by setting $\wh{\bx}_v = +1$ for each $v \in S$ and $\wh{\bx}_v = -1$ for each $v \in \ol{S}$. 

Another reasonable definition would require that the mutual information between $\bX$ and $\wh{\bX}$ is nontrivially large. 
\begin{definition}[Weak recovery --- mutual information]\label{def:mutual-information}
    We say that an algorithm achieves weak recovery in the sense of mutual information if 
    \begin{align*}
        \sum_{v \in [n]} \mathrm{I}(\bX_v ; \wh{\bX}_v) \ge \Omega(n).
    \end{align*}
\end{definition}
By chain rule and monotonicity of entropy, the above guarantee also implies that $\mathrm{I}(\bX; \wh{\bX}) \ge \Omega(n)$.

\subsection{Interpreting our new notion of weak recovery}\label{sec:weak-recovery-intuition}
In this section, we show that if the correlation goes to 1 then we achieve exact recovery on $r$ communities, and if the algorithm uses no information about the graph then we do not achieve $\rho$-weak recovery for any $\rho > 0$ as $n \to \infty$. 

First, let us establish some elementary formulas which will prove useful. Since $\Psi \in \R^{k \times (k-1)}$ is the matrix of right eigenvectors of $T$ with the all-ones vector removed, we have $\Psi \Psi^\top = \Pi^{-1} - \bone_k \bone_k^\top$. We therefore have  
\begin{align*}
    \ev{W_{\Psi}, \bX_{\Psi}} &= \Tr(\Psi^\top W^\top \bX \Psi \Psi^\top \bX^\top W \Psi ) \\
    &= \norm{\Psi^\top W^\top \bX \Psi}_F^2.
\end{align*}
It's not hard to see that conditioned on $\bX$ and $W$, we have 
\begin{align*}
    (W^\top \bX)_{ij}&= n \Pr_{v \sim [n], \wh{\bX}} [\wh{\bX}_v = i | \bX_v = j],
\end{align*}
and below we suppress the subscript in the probabilities for sake of conciseness. Notice that the above formula connects our recovery metric to the \emph{confusion matrix}, a popular metric for multiclass classification in machine learning \cite{DG06}. In particular, the confusion matrix $P \in \R^{k \times k}$ is defined by $P_{ij} \triangleq  \Pr[\wh{\bX}_v = i | \bX_v = j]$, which is exactly the expression on the RHS. 

On the other hand, we can calculate 
\begin{align*}
    \ev{W_{\Psi}, \bX_{\Psi}} &= \ev{W(\Pi^{-1} - \bone_k \bone_k^\top) W^\top, \bX(\Pi^{-1} - \bone_k \bone_k^\top) \bX^\top}\\
    &= \ev{W\Pi^{-1}W^\top, \bX\Pi^{-1}\bX^\top} - \ev{\bone_n \bone_n^\top, \bX \Pi^{-1} \bX^\top + W \Pi^{-1} W^\top} + n^2.
\end{align*}
Now using the fact that $\bone_n^\top \bX = n\pi^\top(1 + \Tilde{O}(1/\sqrt{n}))$, the fact that $\pi$ is a distribution, and writing $\wh{\pi} \triangleq \frac{1}{n}\bone_n^\top W$, we obtain
\begin{align}
    \ev{W_{\Psi}, \bX_{\Psi}} &= \norm{\Pi^{-1/2} W^\top \bX \Pi^{-1/2}}_F^2 - n^2 \norm{\Pi^{-1/2} \wh{\pi}}_2^2 + \Tilde{O}(n^{3/2}) \nonumber \\
    &= n^2 \qty(\sum_{i, j} (\pi(i)\pi(j))^{-1} \Pr[\wh{\bX}_v = i, \bX_v = j]^2 - \norm{\Pi^{-1/2} \wh{\pi}}_2^2 )+ \Tilde{O}(n^{3/2}) \nonumber \\
    &= n^2 \qty(\sum_{i, j} \pi(i)^{-1} \pi(j) \Pr[\wh{\bX}_v = i | \bX_v = j]^2 - \sum_i \pi(i)^{-1} \Pr[\wh{\bX}_v = i]^2 )+ \Tilde{O}(n^{3/2}) \label{eq:correlation-expansion}
\end{align}
with high probability over $\bX$.
\paragraph{Exact recovery for $r$ communities} Suppose there are $r$ eigenvalues of $T$ above the KS threshold. Let $\bU$ be the subspace spanned by $\qty{\bpsi_1^{(n)}, \ldots \bpsi_r^{(n)}}$. If we achieve perfect correlation with $\bU$, then any orthonormal basis $\bB$ of $\bU$ with $\bU = \bB \bB^\top$ exactly recovers at least $r$ of the communities. This follows from the fact that $\bB$ has full row rank, so there are $r$ rows which are pairwise distinct. From these rows we can exactly recover $r$ communities.  
\paragraph{Interpretation of the metric in the symmetric $k$-community block model}
Suppose we are in the symmetric $k$-community block model, with parameter $\lambda$, so that $\pi = \frac{1}{k} \bone$ and $T = k\lambda I + (1-\lambda) \bone \bone^\top$ with $a \neq b$. 
The trivial top eigenvector is $\bone$, and for the remaining eigenvectors we can pick any orthonormal basis in $\bone^\perp$. 
Then $\Psi \Psi^\top = kI - \bone \bone^\top$, and from \cref{eq:correlation-expansion}, when we achieve weak recovery we have
\begin{align*}
    \frac{1}{n^2}\ev{W_{\Psi}, \bX_{\Psi}} &= \sum_{i, j} \qty(\Pr[\wh{\bX}_v = i | \bX_v = j]^2 - \Pr[\wh{\bX}_v = i]^2) + \Tilde{O}(1/\sqrt{n}) \\
    &\ge \Omega(1).
\end{align*}

Therefore, when we succeed at weak recovery in this setting, there exists some $j$ such that 
\begin{align*}
    \sum_{i} \Pr[\wh{\bX}_v = i | \bX_v = j]^2 \ge \sum_{i} \Pr[\wh{\bX}_v = i]^2 + \Omega(1).
\end{align*}
For example, this correlation is maximized if after permuting the columns of $W$ we have $W = X$. 

\paragraph{Comparing to natural baselines}
There are two obvious baselines to sanity check our metric for: labeling vertices randomly according to the prior $\pi$, and deterministically guessing the largest community. In fact, we will characterize the correlation for any randomized rounding scheme which is independent of the true assignment. In this case, we have $W = \bone w^\top$, where $w$ is a simplex vector which is independent of $\bX$. Then we have 
\begin{align*}
    \ev{W_{\Psi}, \bX_{\Psi}} &= \norm{\Psi^\top w \bone^\top \bX \Psi}_F^2 \\
    &= n^2\norm{\Psi^\top w \pi^\top \Psi}_F^2(1 + \Tilde{O}(1/\sqrt{n}))\\
    &= \Tilde{O}(n^{3/2}),
\end{align*}
where the second line holds with high probability over $\bX$, and the third line holds because we removed $\bone$ from $\Psi$. On the other hand, we know that $\norm{W_{\Psi}}_F \norm{\bX_{\Psi}}_F = \Theta(n^2)$, so indeed these baselines fail to achieve weak recovery. 

\paragraph{Indistinguishability and the KS threshold}
Finally, suppose that two communities are actually indistinguishable, in the precise sense that their corresponding rows in $T$ are identical. In our rounding scheme, notice that the rows for $\Psi$ are identical so indeed these two communities are ``merged'' from the perspective of the correlation guarantee. In other words, we could have just reduced the weights $W \in \R^{n \times k}$ into weights $W' \in \R^{n \times (k-1)}$ in the obvious way. 

This is an appealing property because it is a modest step towards understanding the more general setting where two communities are \emph{computationally distinguishable}. 
For example, one could consider a 4-community setup where communities 1 and 2 are computationally indistinguishable (as well as communities 3 and 4), but the pairs of communities $(1,2)$ and $(3,4)$ are computationally distinguishable from each other.

\subsection{Comparing the different notions of weak recovery}\label{sec:weak-recovery-comparison}
In this section, we show that our definition of weak recovery \cref{def:weak-recovery} implies \cref{def:partition,def:mutual-information}.

\begin{lemma}[Weak recovery implies nontrivial partition]
    If a recovery algorithm achieves weak recovery in the sense of \cref{def:weak-recovery}, then it also satisfies \cref{def:partition}.
\end{lemma}
\begin{proof}
    
If $\ev{W_{\Psi}, \bX_{\Psi}} \ge \Omega(n^2)$, then it follows that 
\begin{align*}
    \sum_{c, c' \in [k]} \pi(c)^{-1} \pi(c') \Pr[\wh{\bX}_v = c | \bX_v = c']^2 - \sum_c \pi(c)^{-1} \Pr[\wh{\bX}_v = c]^2 \ge \Omega(1),
\end{align*}
and we conclude that there must exist some $c'$ such that 
\begin{align*}
    \sum_{c} \pi(c)^{-1} (\Pr[\wh{\bX}_v = c | \bX_v = c']^2 - \Pr[\wh{\bX}_v=c]^2) \ge \Omega(1).
\end{align*}

It follows that there must exist some $i \in [k]$ for which $\Pr[\wh{\bX}_v = i | \bX_v = c'] \ge \Pr[\wh{\bX}_v = i] + \Omega(1)$. 
After relabeling the communities for $\bX$, using any $\sigma \in S_k$ which swaps $c'$ and $i$, we can assume that in fact 
\begin{align}
    \Pr[\wh{\bX}_v = i | \bX_v = i] \ge \Pr[\wh{\bX}_v = i] + \Omega(1) \label{eq:marginal-inequality}
\end{align}

Moreover, with $i$ as defined above, since $\Pr[\wh{\bX}_v = i] = \sum_{c} \Pr[\wh{\bX}_v = i | \bX_v = c] \pi(c)$, by averaging there must exist some $j \neq i$  such that 
\begin{align}
    \Pr[\wh{\bX}_v = i | \bX_v = i] \ge \Pr[\wh{\bX}_v = i | \wh{\bX}_v = j] + \Omega(1). \label{eq:nontrivial-cond-probs}
\end{align}
With \cref{eq:nontrivial-cond-probs} in hand, it is easy to see how to define the set $S$ for the partition in \cref{def:partition}. 
In particular, set $S = \qty{v: \hat{X}_v = i}$, and pick the same $i$ and $j$ as in the above guarantee. 
\end{proof}

\begin{lemma}[Weak correlation implies nontrivial mutual information]
    If $\ev{W_{\Psi}, \bX_{\Psi}} \ge \Omega(n^2)$, then $\sum_{v \in [n]} \mathrm{I}(\bX_v ; \wh{\bX}_v) \ge \Omega(n)$.
\end{lemma}
\begin{proof}
    Notice that \cref{eq:marginal-inequality} and another averaging argument implies that there exists a subset $S$ of $\Omega(n)$ vertices such that for each $v \in S$, 
    \begin{align*}
        \Pr_{\wh{\bX}_v, \bX, \bG}[\wh{\bX}_v = i | \bX_v = i] \ge \Pr_{\wh{\bX}_v, \bX, \bG}[\wh{\bX}_v = i] + \Omega(1),
    \end{align*}
    where we have used the subscript to emphasize that the probability no longer samples a random $v \sim [n]$.
    Now by Pinsker's inequality it suffices to lower bound $d_{\mathsf{TV}}(p_{\bX_v, \wh{\bX}_v}, p_{\bX_v} p_{\wh{\bX}_v})$.
    Since $\pi(i) \ge \Omega(1)$, the above inequality implies that $d_{\mathsf{TV}}(p_{\bX_v, \wh{\bX}_v}, p_{\bX_v} p_{\wh{\bX}_v}) \ge \Omega(1)$, so we conclude that 
    \begin{align*}
        \mathrm{I}(\bX_v ; \wh{\bX}_v) &\ge \Omega(1).
    \end{align*}
    Summing up over all $v$ and using the fact that mutual information is nonnegative, we conclude that $\mathrm{I}(\bX; \wh{\bX}) \ge \Omega(n)$.
\end{proof}

%% file: main.bbl
\newcommand{\etalchar}[1]{$^{#1}$}
\begin{thebibliography}{KMM{\etalchar{+}}13}

\bibitem[Abb17]{Abb17}
Emmanuel Abbe.
\newblock Community detection and stochastic block models: recent developments.
\newblock {\em The Journal of Machine Learning Research}, 18(1):6446--6531, 2017.

\bibitem[AS15]{AS15}
Emmanuel Abbe and Colin Sandon.
\newblock Community detection in general stochastic block models: {F}undamental limits and efficient algorithms for recovery.
\newblock In {\em 2015 IEEE 56th Annual Symposium on Foundations of Computer Science}, pages 670--688. IEEE, 2015.

\bibitem[Bas92]{Bas92}
Hyman Bass.
\newblock The {Ihara-Selberg} zeta function of a tree lattice.
\newblock {\em International Journal of Mathematics}, 3(06):717--797, 1992.

\bibitem[BKS23]{BKS23}
Rares-Darius Buhai, Pravesh~K Kothari, and David Steurer.
\newblock Algorithms approaching the threshold for semi-random planted clique.
\newblock In {\em Proceedings of the 55th Annual ACM Symposium on Theory of Computing}, pages 1918--1926, 2023.

\bibitem[BLM15]{BLM15}
Charles Bordenave, Marc Lelarge, and Laurent Massouli{\'e}.
\newblock Non-backtracking spectrum of random graphs: community detection and non-regular {R}amanujan graphs.
\newblock In {\em 2015 IEEE 56th Annual Symposium on Foundations of Computer Science}, pages 1347--1357. IEEE, 2015.

\bibitem[BMR21]{BMR21}
Jess Banks, Sidhanth Mohanty, and Prasad Raghavendra.
\newblock {Local Statistics, Semidefinite Programming, and Community Detection}.
\newblock In {\em Proceedings of the 2021 ACM-SIAM Symposium on Discrete Algorithms (SODA)}, pages 1298--1316. SIAM, 2021.

\bibitem[BS95]{BS95}
Avrim Blum and Joel Spencer.
\newblock Coloring random and semi-random k-colorable graphs.
\newblock {\em Journal of Algorithms}, 19(2):204--234, 1995.

\bibitem[CLMW11]{CLMW11}
Emmanuel~J Cand{\`e}s, Xiaodong Li, Yi~Ma, and John Wright.
\newblock Robust principal component analysis?
\newblock {\em Journal of the ACM (JACM)}, 58(3):1--37, 2011.

\bibitem[CSV17]{CSV17}
Moses Charikar, Jacob Steinhardt, and Gregory Valiant.
\newblock Learning from untrusted data.
\newblock In {\em Proceedings of the 49th Annual ACM SIGACT Symposium on Theory of Computing}, pages 47--60, 2017.

\bibitem[CW04]{CW04}
Moses Charikar and Anthony Wirth.
\newblock Maximizing quadratic programs: Extending grothendieck's inequality.
\newblock In {\em 45th Annual IEEE Symposium on Foundations of Computer Science}, pages 54--60. IEEE, 2004.

\bibitem[DdHS23]{HDdOS23}
Jingqiu Ding, Tommaso d’Orsi, Yiding Hua, and David Steurer.
\newblock Reaching kesten-stigum threshold in the stochastic block model under node corruptions.
\newblock In {\em The Thirty Sixth Annual Conference on Learning Theory}, pages 4044--4071. PMLR, 2023.

\bibitem[DdNS22]{DdNS22}
Jingqiu Ding, Tommaso d'Orsi, Rajai Nasser, and David Steurer.
\newblock Robust recovery for stochastic block models.
\newblock In {\em 2021 IEEE 62nd Annual Symposium on Foundations of Computer Science (FOCS)}, pages 387--394. IEEE, 2022.

\bibitem[DG06]{DG06}
Jesse Davis and Mark Goadrich.
\newblock The relationship between precision-recall and roc curves.
\newblock In {\em Proceedings of the 23rd International Conference on Machine Learning}, pages 233--240, 2006.

\bibitem[DKMZ11]{DKMZ11}
Aurelien Decelle, Florent Krzakala, Cristopher Moore, and Lenka Zdeborov{\'a}.
\newblock Asymptotic analysis of the stochastic block model for modular networks and its algorithmic applications.
\newblock {\em Physical review E}, 84(6):066106, 2011.

\bibitem[DMS17]{DMS17}
Amit Dembo, Andrea Montanari, and Subhabrata Sen.
\newblock Extremal cuts of sparse random graphs.
\newblock {\em The Annals of Probability}, 45(2):1190--1217, 2017.

\bibitem[FK01]{FK01}
Uriel Feige and Joe Kilian.
\newblock Heuristics for semirandom graph problems.
\newblock {\em Journal of Computer and System Sciences}, 63(4):639--671, 2001.

\bibitem[FM17]{FM17}
Zhou Fan and Andrea Montanari.
\newblock How well do local algorithms solve semidefinite programs?
\newblock In {\em Proceedings of the 49th Annual ACM SIGACT Symposium on Theory of Computing}, pages 604--614, 2017.

\bibitem[GHKM23]{GHKM23}
Venkatesan Guruswami, Jun-Ting Hsieh, Pravesh~K Kothari, and Peter Manohar.
\newblock {Efficient Algorithms for Semirandom Planted CSPs at the Refutation Threshold}.
\newblock {\em arXiv preprint arXiv:2309.16897}, 2023.

\bibitem[GP23]{GP23}
Yuzhou Gu and Yury Polyanskiy.
\newblock Uniqueness of {BP} fixed point for the potts model and applications to community detection.
\newblock {\em arXiv preprint arXiv:2303.14688}, 2023.

\bibitem[GV16]{GV16}
Olivier Gu{\'e}don and Roman Vershynin.
\newblock Community detection in sparse networks via {G}rothendieck's inequality.
\newblock {\em Probability Theory and Related Fields}, 165(3-4):1025--1049, 2016.

\bibitem[HS17]{HS17}
Samuel~B Hopkins and David Steurer.
\newblock Efficient {B}ayesian estimation from few samples: community detection and related problems.
\newblock In {\em 2017 IEEE 58th Annual Symposium on Foundations of Computer Science (FOCS)}, pages 379--390. IEEE, 2017.

\bibitem[Iha66]{Iha66}
Yasutaka Ihara.
\newblock On discrete subgroups of the two by two projective linear group over p-adic fields.
\newblock {\em Journal of the Mathematical Society of Japan}, 18(3):219--235, 1966.

\bibitem[KMM{\etalchar{+}}13]{KMMNSZZ13}
Florent Krzakala, Cristopher Moore, Elchanan Mossel, Joe Neeman, Allan Sly, Lenka Zdeborov{\'a}, and Pan Zhang.
\newblock Spectral redemption in clustering sparse networks.
\newblock {\em Proceedings of the National Academy of Sciences}, 110(52):20935--20940, 2013.

\bibitem[KS66]{KS66}
Harry Kesten and Bernt~P Stigum.
\newblock Additional limit theorems for indecomposable multidimensional galton-watson processes.
\newblock {\em The Annals of Mathematical Statistics}, 37(6):1463--1481, 1966.

\bibitem[KS67]{KS67}
Harry Kesten and Bernt~P Stigum.
\newblock Limit theorems for decomposable multi-dimensional galton-watson processes.
\newblock {\em Journal of Mathematical Analysis and Applications}, 17(2):309--338, 1967.

\bibitem[LM22]{LM22}
Allen Liu and Ankur Moitra.
\newblock Minimax rates for robust community detection.
\newblock In {\em 2022 IEEE 63rd Annual Symposium on Foundations of Computer Science (FOCS)}, pages 823--831. IEEE, 2022.

\bibitem[Mas14]{Mas14}
Laurent Massouli{\'e}.
\newblock Community detection thresholds and the weak ramanujan property.
\newblock In {\em Proceedings of the forty-sixth annual ACM symposium on Theory of computing}, pages 694--703, 2014.

\bibitem[MMT20]{MMT20}
Theo McKenzie, Hermish Mehta, and Luca Trevisan.
\newblock A new algorithm for the robust semi-random independent set problem.
\newblock In {\em Proceedings of the Fourteenth Annual ACM-SIAM Symposium on Discrete Algorithms}, pages 738--746. SIAM, 2020.

\bibitem[MNS14]{MNS14}
Elchanan Mossel, Joe Neeman, and Allan Sly.
\newblock Belief propagation, robust reconstruction and optimal recovery of block models.
\newblock In {\em Conference on Learning Theory}, pages 356--370. PMLR, 2014.

\bibitem[MNS18]{MNS18}
Elchanan Mossel, Joe Neeman, and Allan Sly.
\newblock A proof of the block model threshold conjecture.
\newblock {\em Combinatorica}, 38(3):665--708, 2018.

\bibitem[MPW16]{MPW16}
Ankur Moitra, William Perry, and Alexander~S Wein.
\newblock How robust are reconstruction thresholds for community detection?
\newblock In {\em Proceedings of the forty-eighth annual ACM symposium on Theory of Computing}, pages 828--841, 2016.

\bibitem[MS16]{MS16}
Andrea Montanari and Subhabrata Sen.
\newblock Semidefinite programs on sparse random graphs and their application to community detection.
\newblock In {\em Proceedings of the forty-eighth annual ACM symposium on Theory of Computing}, pages 814--827, 2016.

\bibitem[SKZ14]{SKZ14}
Alaa Saade, Florent Krzakala, and Lenka Zdeborov{\'a}.
\newblock Spectral clustering of graphs with the bethe hessian.
\newblock {\em Advances in neural information processing systems}, 27, 2014.

\bibitem[SM19]{SM19}
Ludovic Stephan and Laurent Massouli{\'e}.
\newblock Robustness of spectral methods for community detection.
\newblock In {\em Conference on Learning Theory}, pages 2831--2860. PMLR, 2019.

\bibitem[YP23]{PY23}
Qian Yu and Yury Polyanskiy.
\newblock Ising model on locally tree-like graphs: Uniqueness of solutions to cavity equations.
\newblock {\em IEEE Transactions on Information Theory}, 2023.

\end{thebibliography}
